\documentclass[12pt]{article}
\usepackage[margin=1.0in]{geometry}

\usepackage{amsthm,amsmath,amssymb, amsfonts}
\usepackage{color}
\usepackage{graphicx}

\usepackage[T1]{fontenc}
\usepackage{charter}

\usepackage{mathtools}

\usepackage{setspace}
\def\E{\mathbb{E}}
\def\bF{\mathbb{F}}

\def\bP{\mathbb P}

\def\R{\mathbb{R}}

\def\sA{{\mathcal A}}

\def\sF{{\mathcal F}}

\def\sI{{\mathcal I}}

\def\thup{\theta^\uparrow}          
\def\thupj{\theta^{(j) \uparrow}}   
\def\thdn{\theta^\downarrow}        
\def\thdnj{\theta^{(j) \downarrow}} 
\newcommand{\paren}[1] {^{(#1)}}
\newcommand{\sign}{\text{sign}}

\DeclareMathOperator*{\Argmax}{Arg\,max}

\numberwithin{equation}{section}
\theoremstyle{plain}                
\newtheorem{theorem}{Theorem}[section]
\newtheorem{lemma}[theorem]{Lemma}
\newtheorem{proposition}[theorem]{Proposition}
\newtheorem{corollary}[theorem]{Corollary}
\theoremstyle{definition}           
\newtheorem{definition}[theorem]{Definition}

\theoremstyle{remark}               

\begin{document}

\pagenumbering{arabic} \pagestyle{plain}
\begin{center}
  {\large \bf A multi-agent targeted trading equilibrium with transaction costs}
  \ \\ \ \\
Jin Hyuk Choi, Jetlir Duraj, and Kim Weston\footnote{The first author is supported by the National Research Foundation of Korea (NRF) grant funded by the Korea government (MSIT) (No. 2020R1C1C1A01014142 and No. 2021R1A4A1032924).  The second and third authors acknowledge support by the National Science Foundation under Grant No. DMS\#1908255 (2019-2022). Any opinions, findings and conclusions or recommendations expressed in this material are those of the authors and do not necessarily reflect the views of the National Science Foundation.}
\ \\

\today
\end{center}
\vskip .5in

\begin{abstract}
  We prove the existence of a continuous-time Radner equilibrium with multiple agents and transaction costs.  The agents are incentivized to trade towards a targeted number of shares throughout the trading period and seek to maximize their expected wealth minus a penalty for deviating from their targets.  Their wealth is further reduced by transaction costs that are proportional to the number of stock shares traded.  The agents' targeted number of shares is publicly known, making the resulting equilibrium fully revealing.  In equilibrium, each agent optimally chooses to trade for an initial time interval before stopping trade.  Our equilibrium construction and analysis involves identifying the order in which the agents stop trade.  The transaction cost level impacts the equilibrium stock price drift.  We analyze the equilibrium outcomes and provide numerical examples.
\end{abstract}


\section{Introduction}
We provide the first equilibrium existence result for a continuous-time Radner equilibrium with proportional transaction costs and an arbitrary, finite number of agents.  The agents seek to acquire (or sell) shares of a stock to achieve trading targets by the end of the trading period in a time-weighted average price (TWAP) fashion.  Their wealth declines with trade because of transaction costs that are proportional to the number of shares traded.  The agents seek to maximize their expected wealth minus a TWAP-related penalty term.  The agents' trading targets are publicly known, which leads to a fully revealing equilibrium.

Due to the inherent difficulty of studying models with proportional transaction costs, previous equilibrium existence results are limited in scope.  Previous existence results fall into two categories:  stylized models with only two agents and models that employ approximation or averaging to accomplish market clearing.  Two-agent economies are convenient in equilibrium because market clearing with two agents dictates that each agent must take opposite trades of the other.  Thus, two-agent economies lead to the study of only one agent's optimization problem.  Weston~\cite{W18MAFE}, Noh and Weston~\cite{NW21MAFE}, and Loewenstein and Qin~\cite{LQ21wp} proved equilibrium existence with proportional transaction costs in stylized models with only two agents.  Equilibria with transaction costs are so complicated that even approximate equilibrium results, such as Gonon et.\,al.~\cite{GMKS21MF} and Herdegen and Muhle-Karbe~\cite{HMK18FS}, crucially rely on a two-agent economy.  Continuum-of-agent models, where market clearing is averaged over infinitely many agents, are studied in Vayanos and Vila~\cite{VV99ET}, Vayanos~\cite{V99RES}, Huang~\cite{H03JET}, and D\'avila~\cite{D21JF}.  Lo et.\,al.~\cite{LMW04JPE} and Buss et.\,al.~\cite{BUV14wp} address equilibrium with transaction costs from numerical perspective and do not prove existence. Our model builds off the model in Noh and Weston~\cite{NW21MAFE}, and we prove equilibrium existence in an economy with proportional transaction costs and an arbitrary, finite number of agents.

Our approach to construct an equilibrium starts with a conjecture that is motivated by the trading behavior seen in Noh and Weston~\cite{NW21MAFE}.  We conjecture that each agent optimally chooses to trade for an initial time interval before ceasing trade without resuming.  The model of Noh and Weston~\cite{NW21MAFE} revealed monotonic trading behavior (always buying or always selling) and a simplification of the agents' first-order conditions.  In this work, our conjecture allows us to formulate candidate equilibrium quantities, including the stock's drift and the optimal trading strategies.  Having more than two agents leads to richer trading behavior and a more complicated stock price drift than Noh and Weston~\cite{NW21MAFE}.  Our proposed equilibrium stock price drift is not monotonic, leading to challenges in proving monotonic trading behavior.  Our analysis centers around proving crucial properties of the proposed equilibrium quantities.  In particular, we order the agents according to when they cease trading.  This {\it rank-based} ordering facilitates the definitions of all equilibrium-related quantities and is crucial for our verification and equilibrium existence arguments.

We show that having more than two agents affects the equilibrium stock price drift when transaction costs are proportional. Larger transaction cost proportions lead to less trading, which leads to changes in the stock price drift.  Our results contrast Noh and Weston~\cite{NW21MAFE}, who found no impact of transaction costs on the stock price drift because of simplifications due to a two-agent economy.  Our results share similarities with the two-agent, ergodic-style equilibrium with transaction costs of Gonon et.\,al.~\cite{GMKS21MF}.  In Gonon et.\,al.~\cite{GMKS21MF}, the stock price drift has a larger range of values as transaction costs increase due to an increase in range of an underlying doubly-reflected Brownian motion state process.  We provide analytic results to describe the changes in the equilibrium drift due to transaction costs.

The outline of our paper is as follows.  In Section~\ref{section:set-up}, we set up our model and state our main result on equilibrium existence, Theorem~\ref{thm:existence-1}.  Section~\ref{section:ex} provides an illustrative example that motivates our definition for a rank-based ordering of the agents.  We construct the rank-based ordering and define several equilibrium-related quantities in Section~\ref{section:definitions}.  Using the rank-based ordering, we also refine and restate our main result as Theorem~\ref{thm:existence-2}.  In Section~\ref{section:proofs}, several properties are proven leading up to the proof of Theorem~\ref{thm:existence-2}.  Finally, Section~\ref{section:outcomes} analyzes the outcomes of equilibrium and provides numerical examples.

\section{Model set-up}\label{section:set-up}

Let $T>0$ be a fixed time horizon, which we think of as one trading day in length.  We work in a continuous-time setting and let $B=\left(B_t\right)_{t\in[0,T]}$ be a Brownian motion on a probability space $(\Omega,\sF,\bP)$.  The market consists of two traded securities:  a bank account and a stock.  The bank account is a financial asset in zero-net supply with a constant zero interest rate.  The stock is in constant net supply with the supply denoted by $n\geq0$.  It pays a dividend of $D$ at time $T$.  The random variable $D$ is measurable with respect to $\sigma\left(B_u\!:\, 0\leq u\leq T\right)$ and $\E[D^2]<\infty$.  We let $\sigma=(\sigma_t)_{t\in[0,T]}$ be the progressively measurable process such that $\E\int_0^T\sigma_{u}^2du<\infty$ and
\begin{equation}\label{eqn:D}
  D= \E[D] + \int_0^T \sigma_{u}dB_u
\end{equation}
guaranteed by the martingale representation theorem.  The equilibrium stock price is an It\^o process that will be determined endogenously in equilibrium and is denoted $S = (S_t)_{t\in[0,T]}$.  The price is constrained at time $T$ so that $S_T = D$.  We assume that prices of all goods are denominated in units of a single consumption good.

A finite number of investors, $i=1,\ldots, I$, trade in the market.   They each seek to maximize expected wealth yet are subjected to inventory penalties throughout the trading period.  Their wealth is further penalized by transaction costs, which are proportional to the rate of trade at the rate $\lambda>0$.  The $I=2$ case is studied in Noh and Weston~\cite{NW21MAFE}.  We are mainly interested in the case when $I\geq 3$.

Each agent $i$ has a targeted number of shares $\tilde a_i$ she wishes to acquire (or sell off) throughout the trading period.  The random variables $\tilde a_1, \ldots, \tilde a_I$ are assumed to satisfy $\E[\tilde a_i^2]<\infty$, $i=1,\ldots, I$, and be independent of the Brownian motion $B$.  The trading targets $\tilde a_1, \ldots, \tilde a_I$ are assumed to be known to all investors at time $0$.  The filtration is given by $\bF= (\sF_t)_{t\in[0,T]}$ where
$$
  \sF_t := \sigma\left(\tilde a_1, \ldots, \tilde a_I, B_u:\, u\in[0,t]\right), \quad t\in[0,T],
$$
and we assume that $\sF=\sF_T$.

  A trading strategy $\theta = (\theta_t)_{t\in[0,T]}$ denotes the number of shares held in the stock.  For $1\leq i\leq I$, we say that $\theta$ is \textit{admissible for agent $i$} if it is adapted to $\bF$, c\`adl\`ag, of finite variation on $[0,T]$ $\bP$-a.s., and satisfies $\E\int_0^T \left(\sigma_{t}\theta_t\right)^2dt <\infty$ and $\E\int_0^T\theta_t^2dt<\infty$.  We write $\sA_i$ to denote the collection of admissible strategies for agent $i$.  Agent $i$ is endowed at the beginning of the trading period with $\theta_{i,0-}$ shares of stock, where $\theta_{i,0-}$ is deterministic (constant).  We normalize the endowed shares in the bank account to zero.  For $\theta\in\sA_i$, we allow for $\theta_0$ to differ from $\theta_{i,0-}$, as the agents may choose to trade a lump sum immediately.
  
  For $1\leq i\leq I$, since $\theta\in\sA_i$ is of finite variation, we decompose $\theta$ into 
  \begin{equation}\label{eqn:decomp}
    \theta_t = \theta_{i,0-} + \theta^\uparrow_t - \theta^\downarrow_t, \ \ \ \ t\in[0,T],
  \end{equation}
  where $\theta^\uparrow, \theta^\downarrow$ are adapted to $\bF$, c\`adl\`ag, nondecreasing, and 
\begin{equation}\label{eqn:trade-condition}
  \{t\in[0,T]:d\theta^\uparrow_t>0\}\cap\{t\in[0,T]:d\theta^\downarrow_t>0\}=\emptyset.
\end{equation}
  A change in trading position is possible at time $0$, and we allow for $\theta^\uparrow_0>0$ or $\theta^\downarrow_0>0$ as long as~\eqref{eqn:trade-condition} holds.

At the close of the trading period, the agents consume their acquired dividends.  They are  subjected through their optimization problems to inventory penalties throughout the trading period.  For $1\leq i\leq I$ and a given $\theta\in\sA_i$, the penalty term, or loss term, for agent $i$ is measured by
\begin{equation}\label{def:L}
  L_{i,T}^\theta:= \frac12\int_0^T \kappa(t)\big(\gamma(t)\left(\tilde a_i-\theta_{i,0-}\right) - \left(\theta_t-\theta_{i,0-}\right)\big)^2dt.
\end{equation}
The function $\kappa:[0,T]\rightarrow (0,\infty)$ describes the intensity of the penalty, while $\gamma:[0,T]\rightarrow[0,1]$ describes the desired intraday trading target trajectory.  All agents share the same deterministic functions $\kappa$ and $\gamma$.  We assume that $\gamma$ is continuous, strictly increasing, $\gamma(0)=0$, and $\gamma(T)=1$.  Our main example is time-weighted average price (TWAP), where the intraday trajectory function is $\gamma^{\text{TWAP}}(t):= t/T$.  We assume that $\kappa$ is strictly positive and continuous.

For $1\leq i\leq I$ and a trading strategy $\theta\in\sA_i$, agent $i$'s wealth at $t$ is given by
\begin{equation}\label{def:wealth}
  X^\theta_{i,t} = \theta_{i,0-}S_{0} + \int_{0}^t \theta_u dS_{u} - \lambda\left(\theta^\uparrow_t + \theta^\downarrow_t\right),
  \quad t\in[0,T].
\end{equation}
We recall that the decomposition of $\theta$ in \eqref{eqn:decomp} allows for $\theta^\uparrow_0$ and $\theta^\downarrow_0$ to differ from zero.  Agent $i$'s objective is
$$
  \E\left[X^\theta_{i,T} - L^\theta_{i,T}\,\vert\,\sF_{0}\right] \longrightarrow \ \text{max}
$$
over $\theta\in\sA_i$, where $L^\theta_{i,T}$ is defined in~\eqref{def:L} and $X^\theta_{i,T}$ in~\eqref{def:wealth}.

In an equilibrium, the stock price is determined so that markets clear when all agents  invest optimally.
\begin{definition} \label{def:eq}
Let $\lambda>0$ be a given transaction cost proportion.  An It\^o process $S = \left(S_t\right)_{t\in[0,T]}$ and trading strategies $\theta_1\in\sA_1, \ldots, \theta_I\in\sA_I$ form an \textit{equilibrium} if
\begin{enumerate}
  \item \textit{Strategies are optimal:} For $1\leq i\leq I$, we have that
  $$
    \E\left[X^{\theta_i}_{i,T} - L^{\theta_i}_{i,T}\,\vert\,\sF_{0}\right] 
    = \sup_{\theta\in\sA_i}\E\left[X^\theta_{i,T} - L^\theta_{i,T}\,\vert\,\sF_{0}\right],
  $$
  where $L^\theta_{i,T}$ is defined in~\eqref{def:L}, $X^\theta_{i,T}$ in~\eqref{def:wealth}, and the stock dynamics correspond to $S$ with $S_T=D$.
  \item \textit{Markets clear:} We have $\sum_{i=1}^I\theta_{i,t} = n$ for all $t\in[0,T]$.
\end{enumerate}
\end{definition}

Market clearing in Definition~\ref{def:eq}(2) requires clearing of the stock market, however Walras' Law holds in our model in that the other markets (money market and real goods) clear as well.  We omit the proof here, and instead refer the reader to the analogous result in Noh and Weston~\cite{NW21MAFE} (Lemma 3.2).

Theorem \ref{thm:existence-1} is the main result of the paper.  
\begin{theorem}\label{thm:existence-1}
  For 
  any transaction cost proportion $\lambda>0$, there exists an equilibrium.
\end{theorem}
Section~\ref{section:ex} motivates our equilibrium construction, and Sections~\ref{section:definitions} and \ref{section:proofs} provide the details on the equilibrium construction, its properties, and the proof of Theorem~\ref{thm:existence-1}.

\section{An illustrative example}\label{section:ex}
Motivated by the trading behavior in Noh and Weston~\cite{NW21MAFE}, we formulated the conjecture that for every agent $1\leq i\leq I$, there exists an $\sF_0$-measurable stop-trade time $\tau_i$ such that agent $i$ trades monotonically on $[0, \tau_i]$ and does not trade on $(\tau_i, T]$.  In Sections~\ref{section:definitions} and \ref{section:proofs} we construct equilibrium quantities based on this conjecture and prove that the conjecture holds.  The construction and proof require heavy notation and that several intricate properties hold.  The example in this section provides evidence that the simple idea behind our conjecture is reasonable.

We consider an example with $I=20$ agents, trading day length $T=1$, transaction cost proportion $\lambda = 0.2$, and functions $\kappa(t) = 0.1$ and $\gamma(t) = t$, $t\in[0, 1]$.  We choose the stock supply to be $n=0$ and the dividend as $D$ any random variable with $\E[D^2]<\infty$ and $\E[D]=0$.  The agents are endowed with $\theta_{1,0-} = \ldots = \theta_{I,0-} = 0$ shares initially and have trading targets given by
\begin{align}
  \tilde a_1 &= -300 & \tilde a_6 &= -60 & \tilde a_{11}&=6 & \tilde a_{16}&=70 \nonumber\\
  \tilde a_2 &= -202 & \tilde a_{7} &= -35 & \tilde a_{12}&=11 & \tilde a_{17}&=115 \nonumber\\
  \tilde a_3 &= -165 & \tilde a_{8} &= -20 & \tilde a_{13}&=23 & \tilde a_{18}&=150 \label{ai_data}\\
  \tilde a_4 &= -102 & \tilde a_{9} &= -15 & \tilde a_{14}&=30 & \tilde a_{19} &=220 \nonumber\\
  \tilde a_5 &= -75 & \tilde a_{10} &= 0 & \tilde a_{15}&=63 & \tilde a_{20} &=290. \nonumber
\end{align}  

We conjecture that for each $1\leq i \leq I=20$, there exists an $\sF_0$-measureable stop-trade time $\tau_i$ such that agent $i$ trades on $[0,\tau_i]$ and does not trade at times $(\tau_i, T]$.  Using this conjecture, Figure~\ref{figure:trade-times} plots the trade intervals of each agent.
\begin{figure}[h]
\begin{center}
\includegraphics[scale=0.8]{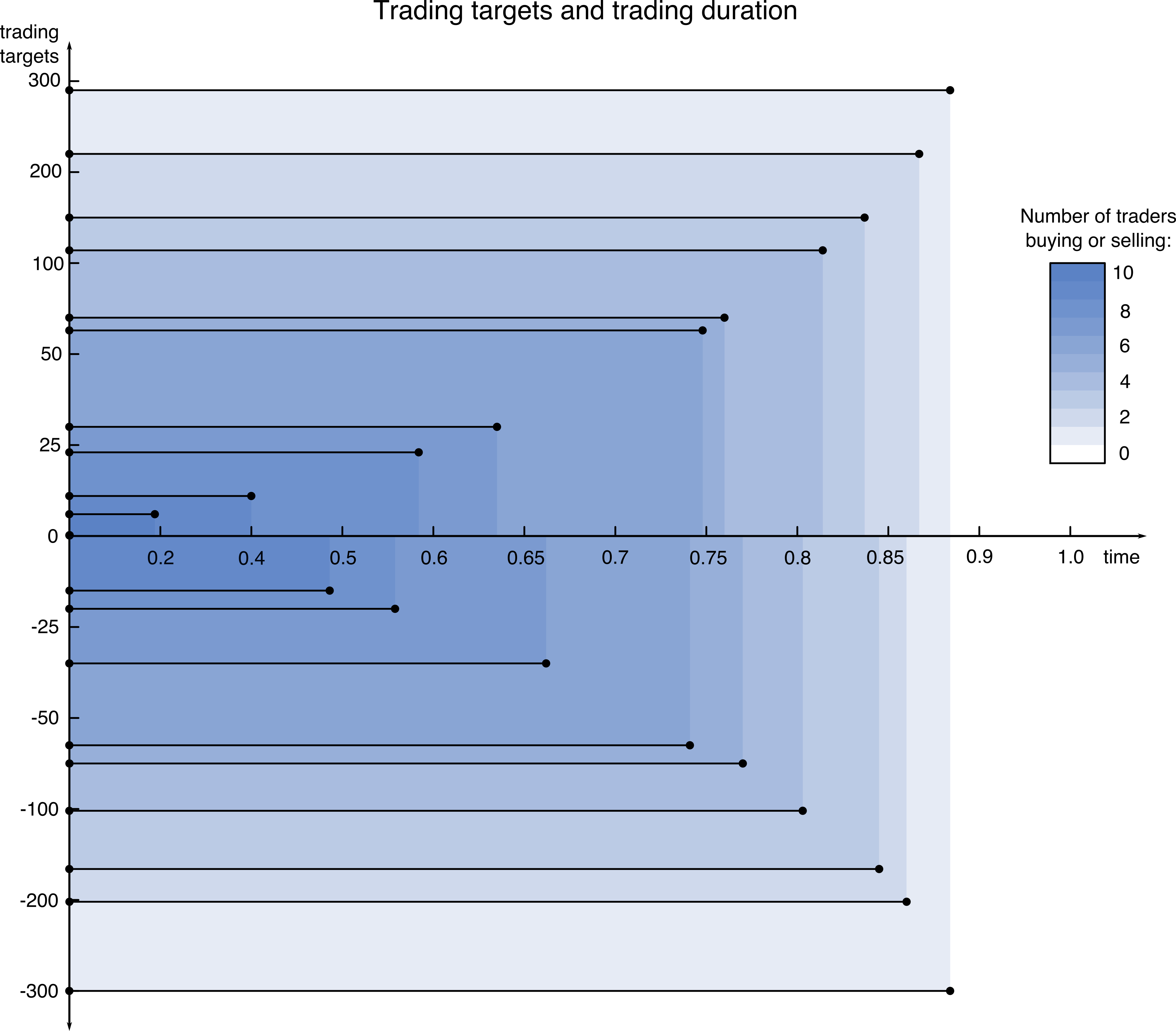}
\end{center}
  \caption{The vertical axis plots the agents' trading targets, while the corresponding length of time trading is plotted on the horizontal axis.  The parameters are $I=20, \, T=1, \, \lambda=0.2,\, n=0,\, \theta_{i,0-}=0$ for $1\leq i \leq I$, $\kappa(t)=0.1$ and $\gamma(t)=\gamma^{\text{TWAP}}(t)=t/T$ for $t\in [0,T]$, and $(\tilde a_i)_{1\leq i \leq I}$ is given in \eqref{ai_data}.  Both the vertical and horizontal axes are scaled to better visualize the agents' behavior.
 }
 \label{figure:trade-times}
\end{figure}

Based on Figure~\ref{figure:trade-times}, the stop-trade times for the agents appear to be ordered, as those agents with targets that are closest to ``average'' target stop trading sooner than the agents with more extreme targets. Agent $10$ with $\tilde a_{10}=0$ chooses never to trade, while the more extreme targets choose to trade for a longer duration. Of course, what is meant by ``average'' must be determined.  It is possible to describe the ordering of the stop-trade times based on an idea of relevant averages, and this construction is detailed in Section~\ref{section:definitions}. We also notice that the last two agents to stop trade are agents $1$ and $20$, who choose to stop trading at the same time, in order to satisfy the market clearing condition.

A necessary property of the stop-trade time conjecture is that trading occurs in a monotone fashion --  an optimizing agent that starts the trading period by buying shares of stock should not want to sell shares of stock later on.  The equilibrium stock price drift appears linearly in the optimal trading strategy formulas and, thus, plays a large role in the monotonicity of the strategies.  Figure~\ref{figure:drift} plots the equilibrium stock price drift.

\begin{figure}[h]
  \begin{center}
  \includegraphics[scale=0.8]{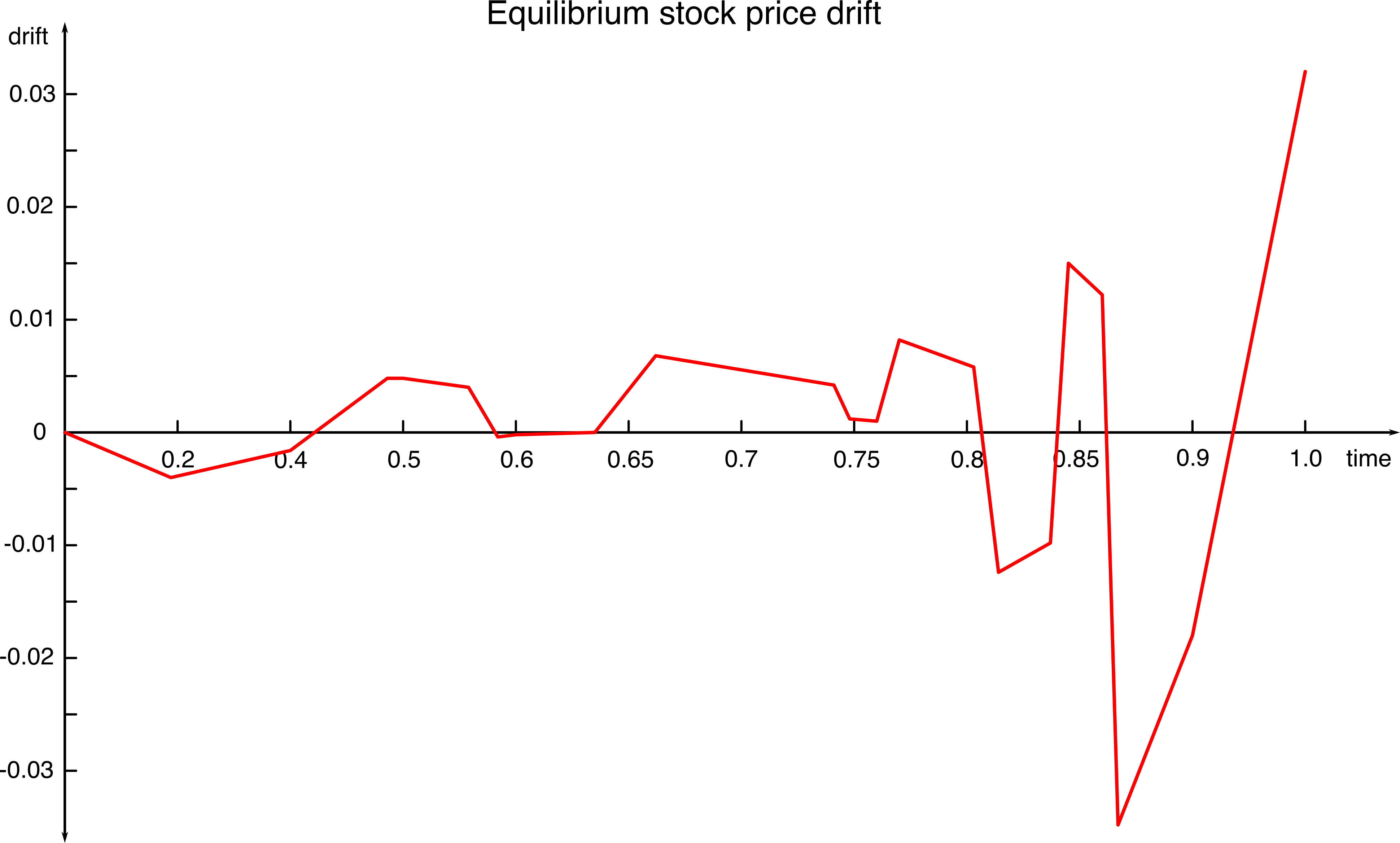}
  \end{center}
  \caption{The equilibrium stock drift is plotted as a function of time. 
The parameters are $I=20, \, T=1, \, \lambda=0.2,\, n=0,\, \theta_{i,0-}=0$ for $1\leq i \leq I$, $\kappa(t)=0.1$ and $\gamma(t)=\gamma^{\text{TWAP}}(t)=t/T$ for $t\in [0,T]$, and $(\tilde a_i)_{1\leq i \leq I}$ is given in \eqref{ai_data}. The horizontal axis is scaled to more suitably visualize the changes in the equilibrium stock drift.
}
\label{figure:drift}
\end{figure}

We see from Figure~\ref{figure:drift} that the equilibrium stock price drift is not necessarily monotone, which creates a challenge for the monotonicity of the optimal trading strategies.  Nevertheless, Figure~\ref{figure:strategies} plots a subset of the optimal trading strategies, which are monotone.

\begin{figure}[h]
\begin{center}
  \includegraphics[scale=0.8]{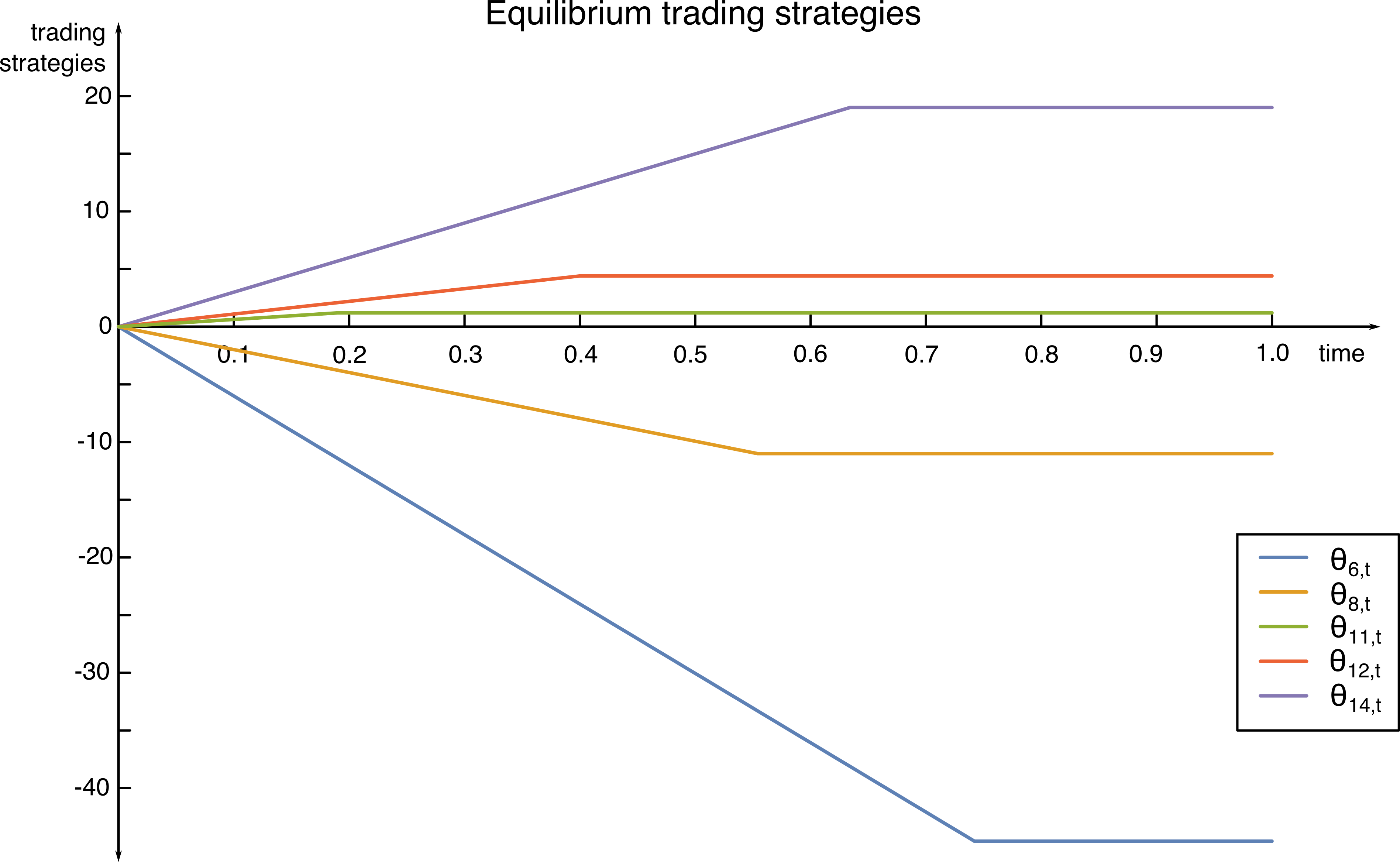}
\end{center}
\caption{The equilibrium trading strategies are plotted as functions of time for agents $6$, $8$, $11$, $12$, and $14$. The parameters are $I=20, \, T=1, \, \lambda=0.2,\, n=0,\, \theta_{i,0-}=0$ for $1\leq i \leq I$, $\kappa(t)=0.1$ and $\gamma(t)=\gamma^{\text{TWAP}}(t)=t/T$ for $t\in [0,T]$, and $(\tilde a_i)_{1\leq i \leq I}$ is given in \eqref{ai_data}.
}
\label{figure:strategies}
\end{figure}

\section{Formal definitions and the rank-based ordering}\label{section:definitions}
Up until now, the agents are labeled in terms of an {\it alphabetical} ordering with index $i\in\{1,\ldots, I\}$ appearing in subscripts.  However, it will be convenient to introduce an alternative {\it rank-based} ordering, which we denote by superscripts in parenthesis. We seek to order the times that agents cease trading by
$$
  0 \leq \tau\paren{1}\leq \ldots\leq \tau\paren{I-1}=\tau\paren{I}<T.
$$
The time $\tau\paren{j}$ represents the $\sF_0$-measurable time when a total of $j$ out of the $I$ agents have stopped trading.  The superscript $(j)$ denotes the $j^{th}$ agent to have dropped out of the market, whereas the subscript $i$ in $\tau_i$ denotes the $\sF_0$-measurable time when agent $i$ has chosen to stop trading.  If agent $i$ is the $j^{th}$ agent to stop trading, then $\tau_i = \tau\paren{j}$.  We will see that it is possible for $\tau\paren{j}=\tau\paren{j+1}$, meaning two agents choose to stop trade at the same time.  Also, we always have $\tau\paren{I-1}=\tau\paren{I}$ since the last two remaining active traders must cease trade at the same time for stock market clearing to hold.

We define the rank-based ordering by backward induction starting with $(I-1)$ and $(I)$.  
In this construction, we will give meaning to the terms: $a\paren{j}$, $a\paren{\geq j}_\Sigma$, and $A\paren{j}$ for trading targets and their aggregate target counterparts; $\tau\paren{j}$ for stop-trade times; and $\sI\paren{j}$ for index sets to translate between alphabetic and rank-based ordering.  After defining these quantities, we define the proposed equilibrium quantities:  $Y\paren{j}$ for candidate optimal first-order conditions; $\theta\paren{j}$ for the candidate optimal trading strategies; and $\mu$ for the stock drift.
\ \\

For convenience, we define $F:[0,T] \to \R$ as
\begin{align}
F(t):=\int_t^T\kappa(u)\big(\gamma(u)-\gamma(t)\big)du. \label{F_def}
\end{align}
%
For $1\leq i\leq I$, we define the relative trading targets by $a_i := \tilde a_i - \theta_{i,0-}$.  The relative trading target $a_i$ represents how many shares agent $i$ must obtain (or sell), relative to her initial position.  As we'll see below, the terms $a_i$ are more relevant to the rank-based definitions and results than the targets $\tilde a_i$.

\noindent{\bf Base case: $j\in\{I-1,I\}$.} We define the initial index set by $\sI\paren{I-1}:=\{1,\ldots, I\}$.  For $i,l\in\sI\paren{I-1}$, we put
\begin{equation}\label{def:eta-last}
  \eta\paren{I-1}_{i,l} := \inf\left\{ t\in[0,T]: \left|\left(a_i-\frac12(a_l+a_i)\right)F(t)\right| \leq\lambda\right\}.
\end{equation}
We choose $i^*, l^*$ to maximize $\eta\paren{I-1}_{i,l}$ so that
$$
  (i^*,l^*) \in\Argmax\left\{\eta\paren{I-1}_{i,l}:\ i,l\in\sI\paren{I-1}\right\}.
$$

We put
\begin{align*}
  \tau\paren{I-1}&:=\tau\paren{I}:=\eta\paren{I-1}_{i^*,l^*}, \\
  \tilde a\paren{I-1}&:=\tilde a_{i^*}, & \tilde a\paren{I}&:=\tilde a_{l^*}, \\
  \theta\paren{I-1}_{0-} &:= \theta_{i^*, 0-}, & \theta\paren{I}_{0-} &:= \theta_{l^*, 0-},\\
  a\paren{I-1}&:= a_{i^*} =  \tilde a\paren{I-1} - \theta\paren{I-1}_{0-}, &  a\paren{I}&:= a_{l^*} = \tilde a\paren{I} - \theta\paren{I}_{0-}, \\
  a\paren{\geq I-1}_\Sigma &:= a\paren{I-1}+a\paren{I},\\
  A\paren{I-1} &:= a\paren{I-1}-\frac12 a\paren{\geq I-1}_\Sigma,  
  & A\paren{I}&:=a\paren{I}-\frac12a\paren{\geq I-1}_\Sigma = -A\paren{I-1}.
\end{align*}

We call the values $\tau\paren{I-1}$ and $\tau\paren{I}$ stop-trade times, which are the $\sF_0$-measurable times that the rank-based agents $(I-1)$ and $(I)$ stop trading.  We also denote agent $(I-1)$'s and $(I)$'s set of admissible strategies by $\sA\paren{I-1}:=\sA_{i^*}$ and $\sA\paren{I}:=\sA_{l^*}$. Finally, we put $\sI\paren{I-2}:=\sI\paren{I-1}\backslash\{i^*, l^*\}$, which will set us up for the next iteration of this recursive definition by removing alphabetic indices that have already been accounted for.
\ \\

\noindent{\bf Backward induction: } Let $j$ be given with $1\leq j\leq I-2$. We assume that the following have been previously defined:  $\sI\paren{j},\ldots,\sI\paren{I-1}$; $\tau\paren{j+1},\ldots, \tau\paren{I}$; $a\paren{j+1},\ldots, a\paren{I}$; $a\paren{\geq j+1}_\Sigma, \ldots, a\paren{\geq I-1}_\Sigma$; and $A\paren{j+1}, \ldots, A\paren{I-1}$. For $ i\in\sI\paren{j}$, we define
\begin{align}
  \eta_i\paren{j} := \inf&\left\{t\in[0,T]:\ \left|\frac{I-j+1}{I-j}\left(a_i-\frac{1}{I-j+1}\left(a_i+a\paren{\geq j+1}_\Sigma\right)\right)F(t) \right.\right. \nonumber\\
  &\left.\left.\quad\quad\quad\quad\quad\quad\quad\quad\quad\quad\quad\quad\quad\quad\quad\quad + \sum_{k=j+1}^{I-2} \frac{A\paren{k}}{I-k} \, F(t\vee\tau\paren{k})\right|\leq\lambda\right\}. \label{def:sigma}
\end{align}
In the definition of $\eta\paren{j}_i$ and throughout this work, we allow for the possibility that summations are empty.  When the upper summation index is strictly less than the lower index, our convention is that the summation is empty and zero-valued.

We choose $i^*$ to maximize $\eta\paren{j}_i$ by
\begin{align}
  i^*\in\Argmax\left\{\eta\paren{j}_i:\ i\in\sI\paren{j}\right\}, \label{choose_j}
\end{align}
and define 
\begin{equation}
\begin{split}\label{various:def}
  \tau\paren{j}&:=\eta\paren{j}_{i^*},\\
  \tilde a\paren{j}&:= \tilde a_{i^*},\\
  \theta\paren{j}_{0-} &:= \theta_{i^*, 0-},\\
  a\paren{j} &:= \tilde a\paren{j}-\theta\paren{j}_{0-},\\
  a\paren{\geq j}_\Sigma &:=a\paren{j}+a\paren{\geq j+1}_\Sigma,\\
  A\paren{j}&:= a\paren{j}-\frac{1}{I-j+1}a\paren{\geq j}_\Sigma,\\
  \sI\paren{j-1}&:=\sI\paren{j}\backslash\{i^*\}.
\end{split}
\end{equation}
We call $\tau\paren{j}$ a stop-trade time since it is the $\sF_0$-measurable time that the rank-based agent $(j)$ stops trading.  We denote agent $(j)$'s set of admissible trading strategies by $\sA\paren{j}:=\sA_{i^*}$. \\

\noindent{\bf Other rank-based and equilibrium quantities:} 
The agents' penalty terms, wealth processes, and optimization problems were previously defined using the {\it alphabetic} notation, and now we rewrite these definitions using the {\it rank-based} notation.  For $1\leq j \leq I$ and $\theta\in\sA\paren{j}$,  agent $(j)$'s penalty term associated with strategy $\theta$ is given by
\begin{equation*}
  L_{T}^{(j),\theta}:= \frac12\int_0^T \kappa(t)\left(\gamma(t)\left(\tilde a\paren{j}-\theta\paren{j}_{0-}\right) - \left(\theta_t-\theta\paren{j}_{0-}\right)\right)^2dt,
\end{equation*}
and agent $(j)$'s wealth process for strategy $\theta$ is
\begin{equation*}
  X^{(j),\theta}_{t} = \theta\paren{j}_{0-}S_{0} + \int_{0}^t \theta_u dS_{u} - \lambda\left(\theta^\uparrow_t + \theta^\downarrow_t\right),
  \quad t\in[0,T].
\end{equation*}
Finally, agent $(j)$ seeks to optimize
\begin{equation}\label{eqn:optimality}
  \E\left[X^{(j),\theta}_{T} - L^{(j),\theta}_{T}\,\vert\,\sF_{0}\right] \longrightarrow \ \text{max}
\end{equation}
over $\theta\in\sA\paren{j}$.  When rewriting the definitions of the penalty terms, wealth processes, and optimization problems, we have not changed any content beyond the agents' notational labels.  The definition of equilibrium, Definition~\ref{def:eq}, is also readily adapted to the rank-based notation.

For notational purposes, we define $\tau\paren{0}:=0$.  The candidate equilibrium stock price drift $\mu$ is defined by
\begin{equation}\label{def:mu}
  \mu_t := 
  \begin{dcases}
    -\kappa(t)\left(\frac{\gamma(t)a\paren{\geq j+1}_\Sigma}{I-j}+\sum_{k=1}^{j}\frac{\gamma(\tau\paren{k})A\paren{k}}{I-k}\right), &t\in[\tau\paren{j},\tau\paren{j+1}), \ 0\leq j\leq I-2,\\
  -\kappa(t)\left(\frac{\gamma(t)a\paren{\geq I-1}_\Sigma}{2}+\sum_{k=1}^{I-2}\frac{\gamma(\tau\paren{k})A\paren{k}}{I-k}\right),  &t\in[\tau\paren{I-1},T].
  \end{dcases}
\end{equation}
For each $1\leq j\leq I$, the candidate optimal trading strategy for agent $(j)$ is given by
\begin{equation}\label{def:theta}
  \theta\paren{j}_t := 
  \begin{dcases}
    \theta\paren{j}_{0-}+\frac{\mu_t}{\kappa(t)} + \gamma(t)a\paren{j}, & t\in[0,\tau\paren{j}],\\
    \theta\paren{j}_{\tau\paren{j}}, & t\in(\tau\paren{j}, T].
  \end{dcases}
\end{equation}

Theorem~\ref{thm:existence-2} (below) restates Theorem~\ref{thm:existence-1} using the rank-based notation.  The proof of Theorem~\ref{thm:existence-2} will show that $\mu$ in \eqref{def:mu} is the equilibrium stock drift, and $\theta\paren{j}$ is agent $(j)$'s optimal trading strategy.  The proof of Theorem~\ref{thm:existence-2} is provided at the end of Section~\ref{section:proofs}.
\begin{theorem}\label{thm:existence-2}
There exists an equilibrium in which the stock price is given by
\begin{equation}\label{def:S}
  S_t := \E[D] + \int_0^t \sigma_u dB_u - \int_t^T \mu_u du, \quad t\in[0,T],
\end{equation}
and $\theta\paren{1}\in\sA\paren{1}, \ldots, \theta\paren{I}\in\sA\paren{I}$ are the optimal trading strategies.
\end{theorem}

Finally, we work towards the definition of the processes $Y\paren{j}$, which are related to the agents' candidate first-order conditions.  To ease notation, we introduce $\Gamma\paren{j}$ for $1\leq j\leq I$ as
$$
  \Gamma\paren{j}_t:= A\paren{j}\int_{t\vee\tau\paren{j}}^T\kappa(u)\big(\gamma(u)-\gamma(\tau\paren{j})\big)du, \quad t\in[0,T].
$$
Since $\kappa$ is strictly positive and $\gamma$ is strictly increasing, we observe that
\begin{align}
&t\mapsto F(t) \quad \textrm{is strictly decreasing,} \label{int_strictly_decreasing}\\
&\left|   \Gamma\paren{j}_t \right| \leq \left| A\paren{j} \right| F(\tau\paren{j}) \quad \textrm{for}\quad t\in [0,T].  \label{F_Gamma_ineq}
\end{align}
Using this notation, agent $(j)$ for $1\leq j\leq I$ has the candidate first-order condition quantity
\begin{align}
  Y\paren{j}_t := \begin{dcases}
    \frac{I-j+1}{I-j}\Gamma\paren{j}_t + \sum_{k=j+1}^{I-2} \frac{\Gamma\paren{k}_t}{I-k}, & 1\leq j \leq I-2,\\
  \Gamma\paren{j}_t, &j =I-1, I,
 \end{dcases}, \quad t\in[0,T].  \label{eqn:Y-def}
\end{align}
For all $1\leq j\leq I$, we notice that if $\tau\paren{j}\leq\tau\paren{j+1}\leq\ldots\leq\tau\paren{I}$, then $Y\paren{j}_t = Y\paren{j}_{\tau\paren{j}}$ for all $t\in[0,\tau\paren{j}]$.  The process $Y\paren{j}$ is similar to the expression appearing inside \eqref{def:eta-last} and \eqref{def:sigma}, and we will prove below in Theorem~\ref{thm:tau-ordering} that $\left|Y\paren{j}\right|\leq\lambda$.

The following lemma is a convenient representation of $Y\paren{j}$ in \eqref{eqn:Y-def}.  The result holds even without assuming (or proving) that $\left|Y\paren{j}\right|\leq\lambda$ or $\tau\paren{j}\leq\tau\paren{j+1}\leq\ldots\leq\tau\paren{I}$. Indeed, once we prove that the $Y\paren{j}$ processes are bounded by $\lambda$, this lemma will be useful for showing that the times $\tau\paren{j}$ are ordered.
\begin{lemma}\label{lemma:Y-representation} For $1\leq j\leq I-2$, we have
\begin{equation}\label{eqn:eq-1}
  \sum_{k=j+1}^{I-2}\frac{\Gamma\paren{k}}{I-k} 
  = \frac{1}{I-j}\sum_{k=j+1}^{I-2}Y\paren{k}.
\end{equation}
Consequently,
\begin{equation}\label{eqn:eq-2}
  Y\paren{j} = \frac{I-j+1}{I-j}\Gamma\paren{j} 
    + \frac{1}{I-j}\sum_{k=j+1}^{I-2} Y\paren{k}.
\end{equation}
\end{lemma}
\begin{proof} We prove the claim by backward induction on $j$.
\ \\
\noindent{\bf Base case:} For $j=I-2$, the sum in \eqref{eqn:eq-1} is empty, and so \eqref{eqn:eq-1} and \eqref{eqn:eq-2} trivially hold.  For $j=I-3$, 
$$
  \sum_{k=j+1}^{I-2} \frac{\Gamma\paren{k}}{I-k} = \frac{\Gamma\paren{j+1}}{I-j-1} = \frac{Y\paren{j+1}}{I-j} = \frac{1}{I-j}\sum_{k=j+1}^{I-2} Y\paren{k},
$$
as desired.
\ \\

\noindent{\bf Induction Hypothesis:} For some $1\leq j\leq I-3$, suppose that 
$$
  \sum_{k=j+2}^{I-2}\frac{\Gamma\paren{k}}{I-k} 
  = \frac{1}{I-j-1}\sum_{k=j+2}^{I-2}Y\paren{k}.
$$

\noindent{\bf Induction Step:} Since
$
  Y\paren{j+1} = \frac{I-j}{I-j-1}\Gamma\paren{j+1} + \sum_{k=j+2}^{I-2}\frac{\Gamma\paren{k}}{I-k},
$
 we deduce that
$$
  \frac{\Gamma\paren{j+1}}{I-j-1} = \frac{Y\paren{j+1}}{I-j} 
  - \frac{1}{I-j}\sum_{k=j+2}^{I-2}\frac{\Gamma\paren{k}}{I-k}.
$$
Thus, we have
\begin{align*}
  \sum_{k=j+1}^{I-2} \frac{\Gamma\paren{k}}{I-k}
  &= \frac{Y\paren{j+1}}{I-j} + \frac{I-j-1}{I-j} \sum_{k=j+2}^{I-2}\frac{\Gamma\paren{k}}{I-k} \\
  &= \frac{1}{I-j}\sum_{k=j+1}^{I-2} Y\paren{k},
  \quad \text{ by the induction hypothesis}.
\end{align*}
By \eqref{eqn:Y-def}, we see that \eqref{eqn:eq-2} holds too.
\end{proof}

The rank-based definitions lead to an ordering of the stop-trade times and boundedness of the candidate first-order conditions.
\begin{theorem}\label{thm:tau-ordering} We have that
$$
  0\leq\tau\paren{1}\leq \tau\paren{2}\leq \ldots \leq \tau\paren{I-1}=\tau\paren{I}<T,
$$
and for all $1\leq j\leq I$ and $t\in[0,T]$, we have that $\left|Y\paren{j}_t\right|\leq \lambda$.
\end{theorem}
\begin{proof}
  First, we notice that for all $1\leq j\leq I$, we have $0\leq \tau\paren{j}<T$. Since $\tau\paren{I-1}=\tau\paren{I}$ and $Y\paren{I-1}=-Y\paren{I}$, we prove the claim for indices $j$ with $1\leq j\leq I-1$.  We proceed by backward induction on $j$.
  \ \\
  
  \noindent{\bf Base Case:} Beginning with $j=I-2$, we seek to prove that $\tau\paren{I-2}\leq\tau\paren{I-1}$, $\left|Y\paren{I-2}\right|\leq\lambda$, and $\left|Y\paren{I-1}\right|\leq\lambda$.  The definitions of $\tau\paren{I-1}$ and $\tau\paren{I-2}$ imply that $\tau\paren{I-1}$ and $\tau\paren{I-2}$ can be rewritten as
\begin{align*}
  \tau\paren{I-1} = \inf&\left\{t\in[0,T]:\ \left| A\paren{I-1}\right| F(t)\leq\lambda\right\},\\
  \tau\paren{I-2} = \inf&\left\{t\in[0,T]:\ \left|\frac32 A\paren{I-2}\right|F(t)\leq\lambda\right\},
\end{align*}
which imply
\begin{align}
\left| A\paren{I-1}\right| F(\tau\paren{I-1}) \leq \lambda \quad \text{and} \quad  \left|\frac32 A\paren{I-2} \right| F(\tau\paren{I-2}) \leq \lambda. \label{AF_bound}
\end{align}
The definition of $Y\paren{I-1}$ and $Y\paren{I-2}$ in \eqref{eqn:Y-def} are expressed by
\begin{align*}
  Y\paren{I-1} = \Gamma\paren{I-1}, \quad
  Y\paren{I-2} = \frac{3}{2}\Gamma\paren{I-2}.
\end{align*}
The above expression, \eqref{F_Gamma_ineq}, and \eqref{AF_bound} imply that $\left|Y\paren{I-1}\right|\leq\lambda$ and $\left|Y\paren{I-2}\right|\leq\lambda$.  

It remains to show that $\tau\paren{I-2}\leq\tau\paren{I-1}$. By \eqref{int_strictly_decreasing}, it suffices to prove that $\left|\frac32 A\paren{I-2}\right|\leq\left|A\paren{I-1}\right|$.
By the definition of $\tau\paren{I-1}$, if $\tau\paren{I-1}>0$, then we have
\begin{equation}\label{equality-2}
  \left|a\paren{I-2}-\frac{a\paren{I-2}+a\paren{I}}{2}\right|
  \leq \left|a\paren{I-1}-\frac{a\paren{I-1}+a\paren{I}}{2}\right|.
\end{equation}
Then, we have
\begin{align*}
  \left|\frac32 A\paren{I-2}\right|
  &=\left|a\paren{I-2}-\frac{a\paren{I-2}+a\paren{I}}{2}\right|\quad\text{by algebra}\\
  &\leq \left|a\paren{I-1}-\frac{a\paren{I-1}+a\paren{I}}{2}\right| \quad\text{by \eqref{equality-2}} \\
  &= \left|A\paren{I-1}\right|.
\end{align*}
Therefore, $\frac32 \left|A\paren{I-2}\right| \leq \left|A\paren{I-1}\right|$,  which implies that $\tau\paren{I-2}\leq\tau\paren{I-1}$.
\ \\

\noindent{\bf Induction Hypothesis:} For some $1\leq j\leq I-3$, suppose that $\tau\paren{j+1}\leq\ldots\leq\tau\paren{I-1}$ and $\left|Y\paren{j+1}\right|, \ldots, \left|Y\paren{I-1}\right| \leq \lambda$. 
  \ \\
  
  \noindent{\bf Induction Step:} Let $1\leq j\leq I-3$ be given, and suppose that the Induction Hypothesis holds for $j$.  We seek to show that $\tau\paren{j}\leq\tau\paren{j+1}$ and $\left|Y\paren{j}\right|\leq \lambda$. 

First, we show $\tau\paren{j}\leq\tau\paren{j+1}$. The definition of $\tau\paren{j}$ implies that 
$$
  \tau\paren{j} = \inf\left\{t\in[0,T]:\ \left|\frac{I-j+1}{I-j}A\paren{j}F(t)  + \sum_{k=j+1}^{I-2} \frac{A\paren{k}}{I-k} \, F(t\vee\tau\paren{k})\right|\leq\lambda\right\}
$$
Therefore, by the definition of $\tau\paren{j}$ and the induction hypothesis that $\tau\paren{j+1}\leq\ldots\leq\tau\paren{I-1}$, it is enough to show that 
\begin{align}
\left| \frac{I-j+1}{I-j}A\paren{j} F(\tau\paren{j+1}) + \sum_{k=j+1}^{I-2} \frac{A\paren{k}}{I-k}F(\tau\paren{k})       \right| \leq \lambda. \label{tau_j+1_wts}
\end{align}
Simple algebra produces
  \begin{equation}
  \begin{split}\label{equality-1}
  \frac{I-j+1}{I-j}A\paren{j} + \frac{1}{I-j-1}A\paren{j+1}=\frac{I-j}{I-j-1}\tilde A\paren{j+1}, \\ \textrm{where} \quad \tilde A\paren{j+1}:=a\paren{j}-\frac{1}{I-j}\left(a\paren{j}+a\paren{\geq j+2}_\Sigma\right).
  \end{split}
\end{equation}
By the induction hypothesis that $\left|Y\paren{j+1}\right|, \ldots, \left|Y\paren{I-2}\right|\leq\lambda$ and identity \eqref{eqn:eq-1} in Lemma~\ref{lemma:Y-representation}, we have
\begin{align}
\left|\sum_{k=j+2}^{I-2}\frac{A\paren{k}}{I-k}F(\tau\paren{k})\right|=    \left| \sum_{k=j+2}^{I-2}\frac{\Gamma\paren{k}_{\tau\paren{j+1}}}{I-k} \right| = \frac{1}{I-j-1}\left|\sum_{k=j+2}^{I-2} Y\paren{k}_{\tau\paren{j+1}}\right| \leq \frac{I-j-3}{I-j-1}\lambda< \lambda. \label{ineq_Y_j+1}
\end{align}
Using \eqref{equality-1}, we check that \eqref{tau_j+1_wts} is equivalent to
\begin{align}
\left|  \frac{I-j}{I-j-1} \tilde A\paren{j+1}F(\tau\paren{j+1})
   +\sum_{k=j+2}^{I-2}\frac{A\paren{k}}{I-k}F(\tau\paren{k})  \right| \leq \lambda. \label{ineq_i*}
\end{align}
Suppose that \eqref{ineq_i*} is false. Then, by \eqref{ineq_Y_j+1} and \eqref{int_strictly_decreasing}, the following inequality should hold:
\begin{align*}
\left|  \frac{I-j}{I-j-1} \tilde A\paren{j+1}F(t)
   +\sum_{k=j+2}^{I-2} \frac{A\paren{k}}{I-k}F(\tau\paren{k})  \right| > \lambda, \quad \textrm{for}\quad t\in [0,\tau\paren{j+1}].
\end{align*}
The above inequality implies $\eta\paren{j+1}_{i^*}>\tau\paren{j+1}$ where $a_{i^*}=a\paren{j}$, which contradicts the construction of $\tau\paren{j+1}$. Therefore, we conclude that \eqref{ineq_i*} holds, so does \eqref{tau_j+1_wts}. 
    
\medskip
  
Second, we prove that $\left|Y\paren{j}\right|\leq\lambda$ on $[0,T]$. Since we have $\tau\paren{j}\leq\ldots\leq\tau\paren{I-1}$ from the first part, we know that $Y\paren{j}_t=Y\paren{j}_{\tau\paren{j}}$ for $t\in[0,\tau\paren{j}]$. The definition of $\tau\paren{j}$ and $\tau\paren{j}\leq\ldots\leq\tau\paren{I-1}$ imply $\left|Y\paren{j}_{\tau\paren{j}}\right|\leq\lambda$. Therefore, we conclude that $\left|Y\paren{j}_t\right|\leq\lambda$ for $t\in[0,\tau\paren{j}]$. We also have that $Y\paren{j}_T=0$.
  
  It remains to prove that $\left|Y\paren{j}_t\right|\leq\lambda$ for $t\in (\tau\paren{j},T)$. Due to the strict positivity of $\kappa$ and the continuity of $Y\paren{j}$, it is enough to show that
  \begin{equation}\label{eqn:Y-bound}
    \frac{Y\paren{j}_t}{\int_t^T\kappa(u)du} \leq \frac{\lambda}{\int_t^T\kappa(u)du} \quad \text{and} \quad -\frac{Y\paren{j}_t}{\int_t^T\kappa(u)du} \leq \frac{\lambda}{\int_t^T\kappa(u)du}, \quad t\in[\tau\paren{j},T).
  \end{equation}
We show this by noting that \eqref{eqn:Y-bound} holds at $t=\tau\paren{j}$ and by checking that
\begin{align}
\left| \frac{d}{dt}\left[\frac{Y\paren{j}_t}{\int_t^T\kappa(u)du}\right] \right| \leq  \frac{d}{dt}\left[\frac{\lambda}{\int_t^T\kappa(u)du}\right], \quad t\in[\tau\paren{j}, T). \label{scaled_Y_prime_ineq} 
\end{align}
  Since $\kappa$ is strictly positive and continuous on $[0,T]$, the above expressions are well-defined.  

  For $t\in[\tau\paren{j},T)$,  either $t\in[\tau\paren{j+m}, \tau\paren{j+m+1})$ for some $0\leq m\leq I-j-3$ or $t\in[\tau\paren{j+m}, T)$ for $m=I-j-2$.  Then,
  $$
    Y\paren{j}_t 
    = \underbrace{\frac{I-j+1}{I-j}\Gamma\paren{j}_t + \sum_{k=j+1}^{j+m} \frac{\Gamma\paren{k}_t}{I-k}}_{t-\text{dependence}} + \underbrace{\sum_{l=j+m+1}^{I-2}\frac{\Gamma\paren{l}_{\tau\paren{l}}}{I-l}}_{\text{no } t-\text{dependence}},
  $$
  and, moreover,
  \begin{align*} 
    \frac{Y\paren{j}_t}{\int_t^T\kappa(u)du}
    &= \underbrace{\left(\frac{I-j+1}{I-j}A\paren{j} + \sum_{k=j+1}^{j+m} \frac{A\paren{k}}{I-k}\right)\frac{\int_t^T\kappa(u)\big(\gamma(u)-\gamma(\tau\paren{j+m})\big)du}{\int_t^T\kappa(u)du} + \frac{\sum_{l=j+m+1}^{I-2}\frac{\Gamma\paren{l}_{\tau\paren{l}}}{I-l}}{\int_t^T\kappa(u)du}}_{t-\text{dependence}}  \\
    & \quad \quad + \underbrace{\frac{I-j+1}{I-j}A\paren{j}\big(\gamma(\tau\paren{j+m})-\gamma(\tau\paren{j})\big) + \sum_{k=j+1}^{j+m-1}\frac{A\paren{k}}{I-k}\big(\gamma(\tau\paren{j+m})-\gamma(\tau\paren{k})\big)}_{\text{no } t-\text{dependence}} 
  \end{align*}
  As usual, the summations are considered empty and zero-valued if their upper and lower bounds do not overlap.  As side calculations, 
\begin{align}
    \frac{d}{dt}\left[\frac{1}{\int^T_t\kappa(u)du}\right] = \frac{\kappa(t)}{\left(\int^T_t\kappa(u)du\right)^2}, \label{kappa_prime}
\end{align}
  and for any $s\in[0,T]$,
  \begin{align*}
    \frac{d}{dt}&\left[\frac{\int^T_t\kappa(u)\big(\gamma(u)-\gamma(s)\big)du}{\int^T_t\kappa(u)du}\right]\\
    &= \frac{\kappa(t)}{\left(\int^T_t\kappa(u)du\right)^2}\left(-\big(\gamma(t)-\gamma(s)\big)\int^T_t\kappa(u)du + \int^T_t\kappa(u)\big(\gamma(u)-\gamma(s)\big)du\right)\\
    &= \frac{\kappa(t)}{\left(\int^T_t\kappa(u)du\right)^2} F(t).
  \end{align*}
  Therefore, we have
  \begin{align*}\label{eqn:Y-derivative}
    \frac{d}{dt}\left[\frac{Y\paren{j}_t}{\int_t^T\kappa(u)du}\right]
    &= \frac{\kappa(t)}{\left(\int_t^T\kappa(u)du\right)^2}\left( \left(\frac{I-j+1}{I-j}A\paren{j} + \sum_{k=j+1}^{j+m} \frac{A\paren{k}}{I-k}\right)F(t) + \sum_{l=j+m+1}^{I-2}\frac{\Gamma\paren{l}_{\tau\paren{l}}}{I-l} \right).
  \end{align*}
  Care should be taken when considering derivatives calculated at the stop-trade times $\tau\paren{j}, \tau\paren{j+1}, \ldots, \tau\paren{I-2}$, since these times are at boundary points of the calculation intervals.  In fact, a closer inspection reveals that the left and right derivatives of $Y\paren{j}/\int_\cdot^T \kappa$ are the same. 

Now we are ready to prove \eqref{scaled_Y_prime_ineq}, and hence to conclude the induction step. By the above expression and \eqref{kappa_prime}, it suffices to show that
  \begin{equation}\label{suff-cond-2}
    \left| \left(\frac{I-j+1}{I-j}A\paren{j} + \sum_{k=j+1}^{j+m} \frac{A\paren{k}}{I-k}\right)F(t) + \sum_{l=j+m+1}^{I-2}\frac{\Gamma\paren{l}_{\tau\paren{l}}}{I-l}  \right| \leq \lambda.
  \end{equation}       
Since $\tau\paren{j}\leq\tau\paren{j+1}\leq\ldots\leq\tau\paren{I-2}$ and $\Gamma\paren{l}_{\tau\paren{l}}=A\paren{l}F(t\vee\tau\paren{l})$ for $l\geq j+m+1$, the above inequality is equivalent to    
    \begin{equation}\label{suff-cond-3}
    \left| \frac{I-j+1}{I-j}A\paren{j} F(t) + \sum_{k=j+1}^{I-2} \frac{A\paren{k}}{I-k}\,F(t\vee\tau\paren{k}) \right| \leq \lambda.
  \end{equation}
Indeed, the definition of $\tau\paren{j}$ and the inequality $t\geq \tau\paren{j}$ imply that \eqref{suff-cond-3} should hold.  
 
\end{proof}

\section{Verification and related inequalities}\label{section:proofs}
In order to prove that the definitions of Section~\ref{section:definitions} satisfy the definition of equilibrium, we need to provide a verification argument.
The argument will rely on Theorem~\ref{thm:tau-ordering}, the properties of $\theta\paren{j}$ in Proposition~\ref{cor:theta-monotone}, and the representation of $Y\paren{j}$ in Proposition~\ref{prop:Y-representation}.



The following lemmas will help to prove needed properties of the candidate optimal trading strategies in Proposition~\ref{cor:theta-monotone}.
\begin{lemma}\label{theta_sign_lem}
Let $1\leq j\leq I-1$ and $1\leq i\leq j-1$. If $\tau\paren{j}>0$ and $A\paren{j}\geq 0$ ($A\paren{j}\leq 0$, resp.), then 
$Y\paren{j}_t=\lambda$ for $t\in [0,\tau\paren{j}]$ and $a\paren{j} \geq a\paren{i}$ ($Y\paren{j}_t=-\lambda$ for $t\in [0,\tau\paren{j}]$ and $a\paren{j} \leq a\paren{i}$, resp.).
\end{lemma}
\begin{proof}
Without loss of generality, we consider the case of $A\paren{j}\geq 0$, since the $A\paren{j}\leq 0$ is handled analogously. Representation \eqref{eqn:eq-2}, the assumption that $\tau\paren{j}>0$, and $\left|Y\paren{j+1}\right|, \ldots, \left|Y\paren{I-1}\right|\leq\lambda$ imply that $\sign(Y\paren{j}_{\tau\paren{j}}) = \sign(A\paren{j})$, and so $Y\paren{j}_{\tau\paren{j}}=\lambda$.
Since $\tau\paren{j}\leq\tau\paren{j+1}\leq\ldots\leq\tau\paren{I}$ by Theorem~\ref{thm:tau-ordering}, the definition of $Y\paren{j}$ in \eqref{eqn:Y-def} implies $Y\paren{j}_t = Y\paren{j}_{\tau\paren{j}}=\lambda$ for all $t\in[0,\tau\paren{j}]$.

Suppose that there exists $1\leq i \leq j-1$ such that $a\paren{j} < a\paren{i}$. Then, 
\begin{align*}
 \frac{I-j+1}{I-j} \left(a\paren{i}-  \frac{a\paren{i}+a\paren{\geq j+1}_\Sigma}{I-j+1} \right) &= a\paren{i}-  \frac{1}{I-j}a\paren{\geq j+1}_\Sigma \\
 &>a\paren{j}-  \frac{1}{I-j}a\paren{\geq j+1}_\Sigma = \frac{I-j+1}{I-j} A\paren{j} \geq 0.
\end{align*}
 Then, the above inequality implies that for $t\in [0,\tau\paren{j}]$,
\begin{align}\label{eqn:contradiction}
 \frac{I-j+1}{I-j} \left(a\paren{i}-  \frac{a\paren{i}+a\paren{\geq j+1}_\Sigma}{I-j+1} \right) F(t) + \sum_{k=j+1}^{I-2} \frac{A\paren{k}}{I-k}\,F(\tau\paren{k})> Y\paren{j}_{\tau\paren{j}}= \lambda.
\end{align}
The inequality~\eqref{eqn:contradiction} contradicts to the maximality of $\tau\paren{j}$ in $\tau\paren{j}$'s definition.
\end{proof}

\begin{lemma}\label{mu_conti_lem}
For $\mu$ defined in \eqref{def:mu}, the map $t\in[0,T]\mapsto \frac{\mu_t}{\kappa(t)}$ is continuous.
\end{lemma}
\begin{proof}
Since $\gamma$ is continuous, it is enough to check that for $1\leq j \leq I-1$,
\begin{align}
\lim_{t\uparrow \tau\paren{j}} \frac{\mu_t}{\kappa(t)} = \frac{\mu_{\tau\paren{j}}}{\kappa(\tau\paren{j})}. \label{mu_conti_check}
\end{align}
By \eqref{def:mu}, we have
\begin{align}
\lim_{t\uparrow \tau\paren{j}} \frac{\mu_t}{\kappa(t)}
& = 
  -\frac{\gamma(\tau\paren{j})a\paren{\geq j}_\Sigma}{I-j+1} - \sum_{k=1}^{j-1}\frac{\gamma(\tau\paren{k})A\paren{k}}{I-k} \nonumber\\
& =  
\begin{dcases}
 -\frac{\gamma(\tau\paren{j})a\paren{\geq j+1}_\Sigma}{I-j} - \sum_{k=1}^{j}\frac{\gamma(\tau\paren{k})A\paren{k}}{I-k}, &1\leq j \leq I-2,\\
   -\frac{\gamma(\tau\paren{j})a\paren{\geq j}_\Sigma}{I-j+1} - \sum_{k=1}^{j-1}\frac{\gamma(\tau\paren{k})A\paren{k}}{I-k}, & j =I-1.
\end{dcases}, \label{mu_conti_check2}
\end{align}
where the second equality is due to $\frac{a\paren{\geq j}_\Sigma}{I-j+1}= \frac{a\paren{\geq j+1}_\Sigma}{I-j}+ \frac{A\paren{j}}{I-j}$. By \eqref{def:mu} and \eqref{mu_conti_check2}, we conclude that \eqref{mu_conti_check} holds. 
\end{proof}

\begin{proposition}\label{cor:theta-monotone}
For $1\leq j\leq I-1$, $\theta\paren{j}$ is continuous and monotone on $[0,T]$.  If $A\paren{j}\geq 0$ ($A\paren{j}\leq 0$, resp.), then $\theta\paren{j}$ is monotone increasing (decreasing, resp.).
\end{proposition}
\begin{proof}
Let $1\leq j\leq I-1$ be given. The continuity of $\theta\paren{j}$ is given by \eqref{def:theta} and Lemma~\ref{mu_conti_lem}.

 For $0\leq i \leq j-1$, by the definitions of $\mu$ in \eqref{def:mu} and $\theta\paren{j}$ in \eqref{def:theta}, we have that
\begin{equation}\label{eqn:theta-formula}
  \theta\paren{j} = \theta\paren{j}_{0-}+\gamma\left(a\paren{j}-\frac{a\paren{\geq i+1}_\Sigma}{I-i}\right) - \sum_{k=1}^i \frac{1}{I-k} \gamma(\tau\paren{k})A\paren{k}
  \quad \text{ on } [\tau\paren{i}, \tau\paren{i+1}).
\end{equation}
Since $\theta\paren{j}$ is constant on $[\tau\paren{j}, T]$, we prove that $\theta\paren{j}$ is monotone on $[0,\tau\paren{j}]$.
By \eqref{eqn:theta-formula} and $\theta\paren{j}$'s continuity at the stop-trade times, it is enough to show that for any $1\leq j' \leq j-1$, the sign of 
$a\paren{j}-  \frac{a\paren{\geq j+1}_\Sigma}{I-j}$ and
$a\paren{j}-  \frac{a\paren{\geq j'+1}_\Sigma}{I-j'}$ are the same, under the assumption that $\tau\paren{j}>0$. Further assume that 
$A\paren{j}\geq 0$, which is equivalent to assumption that $a\paren{j}-  \frac{a\paren{\geq j+1}_\Sigma}{I-j}\geq 0$. Then by Lemma~\ref{theta_sign_lem}, we have $a\paren{j} \geq a\paren{i}$ for $1\leq i \leq j-1$. Using this, we obtain
\begin{align*}
(I-j')\left(a\paren{j}-  \frac{a\paren{\geq j'+1}_\Sigma}{I-j'}\right)&=
(a\paren{j}-a\paren{j'+1}) + \cdots (a\paren{j}-a\paren{j-1}) + 
(I-j)\left( a\paren{j}-  \frac{a\paren{\geq j+1}_\Sigma}{I-j} \right)\\
&\geq 0.
\end{align*}
Thus, $\theta\paren{j}$ is monotone increasing.  The case when $A\paren{j}\leq 0$ is handled analogously.
\end{proof}

\begin{lemma}\label{Aa_eq_lem}
Let $1\leq j \leq I-2$ and $0\leq m \leq I-j-2$. Then,
\begin{align}
\frac{I-j+1}{I-j}A\paren{j} +\sum_{k=j+1}^{j+m} \frac{A\paren{k}}{I-k}=
a\paren{j} - \frac{a\paren{\geq j+m+1}_\Sigma}{I-j-m}. \label{sum_identity}
\end{align}
\end{lemma}
\begin{proof}
Observe that
$$
  a\paren{j} - \frac{a\paren{\geq j+m+1}_\Sigma}{I-j-m} = a\paren{j}-a\paren{j+m+1}+A\paren{j+m+1},
$$
and 
\begin{align*}
\frac{I-j+1}{I-j}A\paren{j}  &= a\paren{j} - a\paren{j+1} + A\paren{j+1},\\
\frac{I-j}{I-j-1}A\paren{j+1}  &= a\paren{j+1} - a\paren{j+2} + A\paren{j+2},\\
&\cdots\\
\frac{I-j-m+1}{I-j-m}A\paren{j+m}  &= a\paren{j+m} - a\paren{j+m+1} + A\paren{j+m+1}.
\end{align*}
We sum the right and left sides above and obtain the desired result.
\end{proof}

The following proposition establishes key properties that allow us to prove that the $Y\paren{j}$ processes satisfy the first-order conditions of the optimization problem \eqref{eqn:optimality}.
\begin{proposition} \label{prop:Y-representation} For each $1\leq j\leq I$,
  \begin{equation}\label{eqn:Y-representation}
    Y\paren{j}_t = \E\left[\int_t^T \kappa(u)\left(\frac{\mu_u}{\kappa(u)}
      + \gamma(u)\left(\tilde a\paren{j}-\theta\paren{j}_{0-}\right)
      - \left(\theta\paren{j}_u - \theta\paren{j}_{0-}\right)\right)du
      \,|\, \sF_t\right], \quad t\in[0,T],
  \end{equation}
  and
  \begin{equation}\label{eqn:reflection}
  \int_0^T \left(\lambda-Y\paren{j}_{t}\right)d\thupj_{t}
  = \int_0^T \left(\lambda+Y\paren{j}_{t}\right)d\thdnj_{t}
  = 0.
\end{equation}
\end{proposition}
\begin{proof} Let $1\leq j\leq I-2$ be given.  For $j\in \{I-1, I\}$, the calculations are similar, though simpler, and we do not include them here.

For $t\in[\tau\paren{j}, T]$, either $t\in[\tau\paren{j+m}, \tau\paren{j+m+1})$ for some $0\leq m \leq I-j-2$ or $t\in[\tau\paren{j+m}, T]$ for $m=I-j-1$.  In either case, we have
$$
  \theta\paren{j}_t = \theta\paren{j}_{\tau\paren{j}}
  = \theta\paren{j}_{0-} + \frac{I-j+1}{I-j}\gamma(\tau\paren{j})A\paren{j} 
    - \sum_{k=1}^j \frac{\gamma(\tau\paren{k})A\paren{k}}{I-k},
$$
and
$$
  \frac{\mu_t}{\kappa(t)} = -\gamma(t)\frac{a\paren{\geq j+m+1}_\Sigma}{I-j-m}
    - \sum_{k=1}^{j+m} \frac{\gamma(\tau\paren{k})A\paren{k}}{I-k}.
$$
Hence,
\begin{align*}
  \frac{\mu_t}{\kappa(t)} &+ \gamma(t)a\paren{j}+\theta\paren{j}_{0-} - \theta\paren{j}_t \\
  &= \gamma(t) \left(a\paren{j}-\frac{a\paren{\geq j+m+1}_\Sigma}{I-j-m}\right)
    - \sum_{k=j+1}^{j+m} \frac{\gamma(\tau\paren{k})A\paren{k}}{I-k} 
    - \frac{I-j+1}{I-j} \gamma(\tau\paren{j})A\paren{j} \\
  &= \gamma(t) \left(\frac{I-j+1}{I-j}A\paren{j}+\sum_{k=j+1}^{j+m}\frac{A\paren{k}}{I-k}\right)
    - \sum_{k=j+1}^{j+m} \frac{\gamma(\tau\paren{k})A\paren{k}}{I-k} 
    - \frac{I-j+1}{I-j} \gamma(\tau\paren{j})A\paren{j} \\
  &= \frac{I-j+1}{I-j}A\paren{j}\left(\gamma(t)-\gamma(\tau\paren{j})\right)
    + \sum_{k=j+1}^{I-2} \frac{A\paren{k}}{I-k}\left(\gamma(t\vee\tau\paren{k})-\gamma(\tau\paren{k})\right),
\end{align*}
where the second equality is due to Lemma~\ref{Aa_eq_lem} and the third equality is due to $t<\tau\paren{k}$ for $k\geq j+m+1$. 
Therefore, for any $t\in[\tau\paren{j}, T]$,
\begin{align*}
  \int_t^T \kappa(u)&\left(\frac{\mu_u}{\kappa(u)}+\gamma(u)a\paren{j} + \theta\paren{j}_{0-}-\theta\paren{j}_u\right)du \\
  &= \int_t^T\kappa(u)\left(\frac{I-j+1}{I-j}A\paren{j}\left(\gamma(u)-\gamma(\tau\paren{j})\right)
    + \sum_{k=j+1}^{I-2} \frac{A\paren{k}}{I-k}\left(\gamma(u\vee\tau\paren{k})-\gamma(\tau\paren{k})\right)\right)du \\
  &= \frac{I-j+1}{I-j}\Gamma\paren{j}_t + \sum_{k=j+1}^{I-2}\frac{\Gamma\paren{k}_t}{I-k} \\
  &= Y\paren{j}_t,
\end{align*}
Since $Y\paren{j}$ and all terms in $\left(\frac{\mu}{\kappa} + a\paren{j}\gamma + \theta\paren{j}_{0-}-\theta\paren{j}\right)$ are $\sF_0$-measurable, we conclude that \eqref{eqn:Y-representation} holds.

Finally, we prove that the reflection condition~\eqref{eqn:reflection} holds.  Suppose that $\tau\paren{j}=0$. Then, Theorem~\ref{thm:tau-ordering} implies $\tau\paren{k}=0$ for $1\leq k \leq j$. Therefore, \eqref{def:mu}-\eqref{def:theta} and $\gamma(0)=0$ produce $\theta\paren{j}_t=\theta\paren{j}_0=\theta\paren{j}_{0-}$ for $t\in [0,T]$. In this case, \eqref{eqn:reflection} trivially holds.

Suppose now that $\tau\paren{j}>0$ and $A\paren{j}\geq 0$. By Lemma~\ref{theta_sign_lem} and Proposition~\ref{cor:theta-monotone}, we have
\begin{align*}
Y\paren{j}_t &= \lambda \quad \textrm{for} \quad t\in [0,\tau\paren{j}],\\
\theta\paren{j}_t &= \theta\paren{j}_T \quad \textrm{for} \quad t\in [\tau\paren{j},T],\\
\thdnj_{t} &=0 \quad \textrm{for} \quad t\in [0,T].
\end{align*}
The above equations produce \eqref{eqn:reflection}. If $\tau\paren{j}>0$ and $A\paren{j}\leq 0$, then \eqref{eqn:reflection} still holds by the same way.
\end{proof}

Next, we work towards our proof of equilibrium existence by studying the individual agents' optimization problems.  We recall that for $1\leq j\leq I$ and $\theta\in\sA\paren{j}$, agent $(j)$'s loss term is given by 
\begin{equation*}\label{def:L-j}
  L^{(j),\theta}_{T} := \frac12\int_0^T \kappa(t)\left(\gamma(t)\left(\tilde a\paren{j}-\theta\paren{j}_{0-}\right) - \left(\theta_t-\theta\paren{j}_{0-}\right)\right)^2dt,
\end{equation*}
while her wealth at $t$ is given by
\begin{equation*}\label{def:wealth-j}
  X^{(j),\theta}_{t} = \theta\paren{j}_{0-}S_{0} + \int_{0}^t \theta_u dS_{u} - \lambda\left(\theta^\uparrow_t + \theta^\downarrow_t\right),
  \quad t\in[0,T].
\end{equation*}
Agent $(j)$ seeks to maximize
\begin{equation*}\label{optimality}
  V\paren{j}(\theta):= \E\left[X^{(j),\theta}_{T} - L^{(j),\theta}_{T}\,\vert\,\sF_{0}\right]
\end{equation*}
over $\theta\in\sA\paren{j}$.

We also recall that $\sigma=(\sigma_t)_{t\in[0,T]}$ is given by the martingale representation of the terminal dividend $D$ by
$$
  D = \E\left[D\right] + \int_0^T \sigma_u dB_u.
$$

\begin{proposition}\label{prop:theta_optimal}
 Assume that the stock has dynamics
  $$
    dS_t = \mu_t dt + \sigma_t dB_t, \quad S_T = D.
  $$
  For $1\leq j\leq I$, the trading strategy $\theta\paren{j}\in\sA\paren{j}$ is optimal for agent $(j)$ in \eqref{eqn:optimality}.
\end{proposition}
\begin{proof} Let $1\leq j\leq I$ be given. First, we show that $\theta\paren{j}\in\sA\paren{j}$.   Since $\theta\paren{j}$ is $\sF_0$-measurable, it is also $\bF$-adapted.  Proposition~\ref{cor:theta-monotone} proves that $\theta\paren{j}$ is continuous and monotone, so it is c\`adl\`ag and of finite variation.  Finally, $\theta\paren{j}$ is bounded by a constant times $\sup_{1\leq i\leq I}|a_i|$, where we have $\E\left[\left(\sup_{1\leq i\leq I}|a_i|\right)^2\right]<\infty$ by the assumptions that $\E\left[(a_i+\theta_{i,0-})^2\right]<\infty$ and $\theta_{i,0-}$ is constant for all $1\leq i\leq I$.  We recall that by assumption, $\tilde a_i = a_i+\theta_{i,0-}$ is independent of the Brownian motion and $D$ is measurable with respect to the Brownian filtration with $\E[D^2]<\infty$.  Thus, $\theta\paren{j}$ is independent of $\sigma$, $\E\int_0^T\left(\sigma_t\theta\paren{j}_t\right)^2dt <\infty$, $\E\int_0^T\left(\theta\paren{j}_t\right)^2dt<\infty$, and so, $\theta\paren{j}\in\sA\paren{j}$.

Next, we proceed to show that $\theta\paren{j}$ is optimal for agent $(j)$.  Since $\theta, \theta\paren{j}\in\sA\paren{j}$, and by the definition of $\mu$ in \eqref{def:mu}, the integrals and conditional expectations in the calculations below are well-defined. For any $\theta\in\sA\paren{j}$, we have
\begin{align*}
&V\paren{j}(\theta)-V\paren{j}(\theta\paren{j}) \\
&= \lambda\,\E\left[\thupj_{T}+\thdnj_{T}-\thup_{T}-\thdn_{T}\,|\,\sF_{0}\right] \\
&\ \ \ \ \ \ +\E\left[\int_0^T \kappa(u)
\left( \frac12\left((\theta\paren{j}_{u})^2-\theta_u^2\right) 
+ (\theta_u-\theta\paren{j}_{u})\left(\frac{\mu_u}{\kappa(u)} +\gamma(u)a\paren{j}+\theta\paren{j}_{0-}\right)\right)du\,|\,\sF_{0}\right]\\
&= \lambda\,\E\left[\thupj_{T}+\thdnj_{T}-\thup_{T}-\thdn_{T}\,|\,\sF_{0}\right] \\
&\ \ \ \ \ \ +\E\left[\int_0^T \kappa(u)\left(
\left(\theta_u-\theta\paren{j}_{u}\right)\left(\frac{\mu_u}{\kappa(u)}+\gamma(u)a\paren{j} +\theta\paren{j}_{0-}-\theta\paren{j}_{u}\right)-\frac12\left(\theta_u-\theta\paren{j}_{u}\right)^2\right)du
\,|\,\sF_{0}\right],\\
&\leq \lambda\,\E\left[\thupj_{T}+\thdnj_{T}-\thup_{T}-\thdn_{T}
+\frac{1}{\lambda}\int_0^T \kappa(u)
\left(\theta_u-\theta\paren{j}_{u}\right)\left(\frac{\mu_u}{\kappa(u)}+\gamma(u)a\paren{j} +\theta\paren{j}_{0-}-\theta\paren{j}_{u}\right)du\,|\,\sF_{0}\right] \\
&= \lambda\,\E\left[\thupj_{T}+\thdnj_{T}-\thup_{T}-\thdn_{T} - \frac{1}{\lambda}\int_0^T \left(\theta_u-\theta\paren{j}_{u}\right)dY\paren{j}_{u}\,|\,\sF_{0}\right],
\end{align*}
where the last equality is due to \eqref{eqn:Y-representation}. By integration by parts, the above is rewritten as
\begin{align}
& \E\left[\lambda\thupj_{T}+\lambda\thdnj_{T}-\lambda\thup_{T}-\lambda\thdn_{T} -Y\paren{j}_{T}\left(\theta_T-\theta\paren{j}_{T}\right) + Y\paren{j}_{0}\left(\theta\paren{j}_{0-}-\theta\paren{j}_{0-}\right) + \int_{0-}^T Y\paren{j}_{u}d\left(\theta_u-\theta\paren{j}_{u}\right)\,|\,\sF_{0}\right] \nonumber\\
&= \E\left[\int_{0-}^T\left(\lambda-Y\paren{j}_{u}\right)d\thupj_{u}
  +\int_{0-}^T\left(\lambda+Y\paren{j}_{u}\right)d\thdnj_{u} \,\Big| \,\sF_{0}\right] \label{terms-1}\\
&\qquad \qquad\qquad\qquad
+\E\left[  
  \int_{0-}^T\left(Y\paren{j}_{u}-\lambda\right)d\thup_{u}
  +\int_{0-}^T\left(-\lambda-Y\paren{j}_{u}\right)d\thdn_{u} \, \Big| \,\sF_{0}\right],\label{terms-2} 
\end{align}
where we use $Y\paren{j}_{T}=0$ and $\left(\theta\paren{j}_{0-}-\theta\paren{j}_{0-}\right)=0$. The terms in \eqref{terms-1} are zero due to \eqref{eqn:reflection}, and the terms in \eqref{terms-2} are nonpositive due to $\left|Y\paren{j}\right|\leq \lambda$ in Theorem~\ref{thm:tau-ordering}. All in all, we conclude 
\begin{align*}
V\paren{j}(\theta)-V\paren{j}(\theta\paren{j}) \leq 0.
\end{align*}
Therefore, $\theta\paren{j}\in\sA\paren{j}$ is the optimal trading strategy for agent $(j)$.
\end{proof}

We now work towards the proof of our equilibrium existence result, Theorem~\ref{thm:existence-2}.  Lemma~\ref{lem:sum_theta_formula} (below) proves a formula that will be used to prove market clearing.
\begin{lemma}\label{lem:sum_theta_formula}
The following holds:
\begin{align}
\sum_{k=1}^j \theta\paren{k}_{\tau\paren{k}} 
= 
\begin{dcases}
 \sum_{k=1}^{j} \theta\paren{k}_{0-} + (I-j) \sum_{k=1}^{j} \frac{\gamma(\tau\paren{k})A\paren{k}}{I-k}, & 1\leq j \leq I-1, \\
 n, & j=I.
\end{dcases} \label{sum_theta_formula}
\end{align} 
\end{lemma}
\begin{proof}
For $1\leq k\leq I-1$, by the definition of $\theta\paren{k}$ in \eqref{def:theta}, we have
\begin{align}
\theta\paren{k}_{\tau\paren{k}} &= \theta\paren{k}_{0-}+\frac{\mu_{\tau\paren{k}}}{\kappa(\tau\paren{k})} + \gamma(\tau\paren{k})a\paren{k} \nonumber \\
&=
 \theta\paren{k}_{0-} + \gamma(\tau\paren{k}) A\paren{k} - \sum_{i=1}^{k-1} \frac{\gamma(\tau\paren{i})A\paren{i}}{I-i}, \label{theta_k_formula}
\end{align}
where we use Lemma~\ref{mu_conti_lem} to obtain the second equality. We can easily see that \eqref{sum_theta_formula} holds for $j=1$ due to \eqref{theta_k_formula} with $k=1$. Now, as an induction hypothesis, suppose that \eqref{sum_theta_formula} holds for $1\leq j \leq I-2$. Then, by using \eqref{theta_k_formula}, we obtain
\begin{align}
\sum_{k=1}^{j+1} \theta\paren{k}_{\tau\paren{k}} 
&= \theta\paren{j+1}_{\tau\paren{j+1}} +   \sum_{k=1}^{j} \theta\paren{k}_{0-} + (I-j) \sum_{k=1}^{j} \frac{\gamma(\tau\paren{k})A\paren{k}}{I-k}\\
&= \sum_{k=1}^{j+1} \theta\paren{k}_{0-} + (I-j-1) \sum_{k=1}^{j+1} \frac{\gamma(\tau\paren{k})A\paren{k}}{I-k}.
\end{align}
Therefore, by induction, we conclude that \eqref{sum_theta_formula} holds for $1\leq j \leq I-1$.

It only remains to prove \eqref{sum_theta_formula} for $j=I$. By \eqref{def:theta}, we have
\begin{align}
\theta\paren{I}_{\tau\paren{I}} &= \theta\paren{I}_{0-}+\frac{\mu_{\tau\paren{I}}}{\kappa(\tau\paren{I})} + \gamma(\tau\paren{I})a\paren{I} \nonumber \\
&=
 \theta\paren{I}_{0-} + \gamma(\tau\paren{I}) A\paren{I} - \sum_{k=1}^{I-2} \frac{\gamma(\tau\paren{k})A\paren{k}}{I-k}, \label{theta_I_formula}
\end{align}
Finally, \eqref{sum_theta_formula} with $j=I-1$ and \eqref{theta_I_formula} complete the proof:
\begin{align*}
\sum_{k=1}^I \theta\paren{k}_{\tau\paren{k}}  
&=  \left(\sum_{k=1}^{I-1} \theta\paren{k}_{0-} + \sum_{k=1}^{I-1} \frac{\gamma(\tau\paren{k})A\paren{k}}{I-k} \right) +
\left(
 \theta\paren{I}_{0-} + \gamma(\tau\paren{I}) A\paren{I} - \sum_{k=1}^{I-2} \frac{\gamma(\tau\paren{k})A\paren{k}}{I-k}
\right)\\
&=  \sum_{k=1}^{I} \theta\paren{k}_{0-} + \gamma(\tau\paren{I-1}) A\paren{I-1} + \gamma(\tau\paren{I}) A\paren{I} \\
&=n,
\end{align*}
where the last equality is due to $\tau\paren{I-1}=\tau\paren{I}$ and $A\paren{I-1} =-A\paren{I}$.
\end{proof}

Finally, we have all of the pieces in place to prove the equilibrium existence result, Theorem~\ref{thm:existence-2}.
\begin{proof}[Proof of Theorem~\ref{thm:existence-2}]
First, we notice from the definition of $S$ in \eqref{def:S} that $S_T = \E[D]+\int_0^T \sigma_u dB_u = D$. Since the optimality of $\theta\paren{j}$ is already checked in Proposition~\ref{prop:theta_optimal}, to complete the proof, it only remains to check the market clearing condition:
\begin{align}
\sum_{k=1}^I \theta\paren{k}_{t} =n, \quad t\in [0,T].
\end{align}

Suppose that $0\leq j \leq I-1$ and $t\in[\tau\paren{j},\tau\paren{j+1})$. By \eqref{def:mu} and \eqref{def:theta}, we have
\begin{align}
\theta\paren{k}_t = \theta\paren{k}_{0-} - \gamma(t) \frac{a\paren{\geq j+1}_\Sigma}{I-j} - \sum_{i=1}^j  \frac{\gamma(\tau\paren{i})A\paren{i}}{I-i} + \gamma(t) a\paren{k}, \quad j+1\leq k \leq I. \label{theta_k_t_formula}
\end{align}
Then, we obtain
\begin{align*}
\sum_{k=1}^I \theta\paren{k}_{t}  &= \sum_{k=1}^j \theta\paren{k}_{\tau\paren{k}} + \sum_{k=j+1}^I \theta\paren{k}_t \\
&=  \sum_{k=1}^{I} \theta\paren{k}_{0-} + (I-j) \sum_{k=1}^{j} \frac{\gamma(\tau\paren{k})A\paren{k}}{I-k}    
+ \gamma(t) \sum_{k=j+1}^I a\paren{k} - \gamma(t) a\paren{\geq j+1}_\Sigma - (I-j)\sum_{i=1}^j \frac{\gamma(\tau\paren{i})A\paren{i}}{I-i}\\
&= \sum_{k=1}^{I} \theta\paren{k}_{0-} = n, 
\end{align*}
where the first equality is due to \eqref{def:theta} and Theorem~\ref{thm:tau-ordering}, and the second equality is due to Lemma~\ref{lem:sum_theta_formula} and \eqref{theta_k_t_formula}.
  
Finally, for $t\in [\tau\paren{I-1},T]$, we have
\begin{align}
\sum_{k=1}^I \theta\paren{k}_{t}  = \sum_{k=1}^I \theta\paren{k}_{\tau\paren{k}}= n,  \label{clear_last_interval} 
\end{align}
where the first equality is due to \eqref{def:theta} and Theorem~\ref{thm:tau-ordering}, and the second equality is due to Lemma~\ref{lem:sum_theta_formula}. 
\end{proof}

\section{Analyzing equilibrium outcomes}\label{section:outcomes}
How does the equilibrium stock drift $\mu$ depend on the transaction cost $\lambda$?  In Noh and Weston~\cite{NW21MAFE}, the two-agent economy means that $\mu$ does not vary with $\lambda$.  In Gonon et.\,al.~\cite{GMKS21MF}'s ergordic-style equilibrium, the equilibrium drift is impacted by $\lambda$, even with a two-agent economy because the two agents have different quadratic penalty term coefficients. In our work, the presence of multiple agents gives way to $\mu$'s dependence on $\lambda$, which is intricate and investigated below.

 First, we  prove that the rank-based ordering $a\paren{j}$ only depends on the relative trading targets, $(a_i)_{1\leq i\leq I}$, not on $\gamma, \kappa, \lambda$.
\begin{proposition}\label{prop:rank-based}
 Let $1\leq k\leq I$ and $\tau\paren{k}>0$. Then, 
 \begin{enumerate}
   \item $A\paren{k}F(\tau\paren{k})=c_k  \lambda$, where $c_k$ is a constant that only depends on $(a_i)_{1\leq i\leq I}$, 
    not on $\gamma, \kappa, \lambda$.
   \item $a\paren{k}$ only depends on $(a_i)_{1\leq i\leq I}$, not on $\gamma$, $\kappa$, or $\lambda$.
   \end{enumerate}
 \end{proposition}
 \begin{proof}
  We prove this by backward induction. The cases $k=I, I-1$ are simple, so we start with $k=I-2$.
  Recall that if $\tau\paren{I-2}>0$, then $\tau\paren{I-1}>0$ and $a\paren{\geq I-1}_\Sigma=\max_{1\leq i \leq I}\{ a_i \} + \min_{1\leq i \leq I} \{ a_i \}$. Also,
  $$a\paren{I-2}=\Argmax \left\{a_i: \left| a_i - \frac{a\paren{\geq I-1}_\Sigma}{2} \right|, i\in\sI\paren{I-2}\right\},
  $$
  which implies that $a\paren{I-2}, a\paren{\geq I-1}_\Sigma, A\paren{I-2}$ only depend on $(a_i)_{1\leq i\leq I}$. The first statement also holds by
 $$
 A\paren{I-2}F(\tau\paren{I-2})=\sign(A\paren{I-2})\frac{2}{3}\lambda.
 $$

As the induction hypothesis, suppose that the statements hold for $j+1\leq k \leq I$ and $\tau\paren{j}>0$. Recall that \eqref{eqn:eq-1} and Theorem~\ref{thm:tau-ordering} imply $\left| \sum_{k=j+1}^{I-2} \frac{c_k}{I-k}\right|<1$. Therefore, if $\eta\paren{j}_i>0$, then
  \begin{align*}
\left( a_i- \frac{a\paren{\geq j+1}_\Sigma}{I-j} \right)F(\eta\paren{j}_i) + \lambda\sum_{k=j+1}^{I-2} \frac{c_k}{I-k} = 
 \begin{dcases}
 \lambda, &\textrm{if  }\left( a_i- \frac{a\paren{\geq j+1}_\Sigma}{I-j} \right)>0\\
 -\lambda, &\textrm{if  }\left( a_i- \frac{a\paren{\geq j+1}_\Sigma}{I-j} \right)<0
 \end{dcases} 
 \end{align*}
 Considering $\tau\paren{j}=\max_{i\in \sI\paren{j}}\{\eta\paren{j}_i\}$, we can see that $a\paren{j}$ is chosen as the value of $a_i$ for $i\in\sI\paren{j}$ that maximizes the expression:
 $$
   \frac{1}{1-\sum_{k=j+1}^{I-2} \frac{c_k}{I-k}}\left( a_i- \frac{a\paren{\geq j+1}_\Sigma}{I-j} \right)^+ +\frac{1}{1+\sum_{k=j+1}^{I-2} \frac{c_k}{I-k}} \left( a_i- \frac{a\paren{\geq j+1}_\Sigma}{I-j} \right)^-,
 $$
 which depends only on $(a_i)_{1\leq i\leq I}$, 
  and not on $\gamma$, $\kappa$, or $\lambda$.  Finally, 
  \begin{align*}
 \frac{I-j+1}{I-j}A\paren{j} F(\tau\paren{j}) + \lambda\sum_{k=j+1}^{I-2} \frac{c_k}{I-k} = 
 \begin{dcases}
 \lambda, &\textrm{if  }A\paren{j}>0\\
 -\lambda, &\textrm{if  }A\paren{j}<0
 \end{dcases} 
 \end{align*}
produces
  \begin{align*}
c_j=    \frac{I-j}{I-j+1} \left(\sign(A\paren{j})-\sum_{k=j+1}^{I-2} \frac{c_k}{I-k}\right),
 \end{align*}
 and $c_j$ also only depends on $(a_i)_{1\leq i\leq I}$.
  \end{proof}

In Section 4, we order $(a_i)_{1\leq i\leq I}$ by the backward induction \eqref{def:sigma}-\eqref{various:def} and obtain $(a\paren{j})_{1\leq j\leq I}$. The resulting $(a\paren{j})_{1\leq j\leq I}$ may not be unique: if $  \{i_1, i_2 \}\subset \Argmax\left\{\eta\paren{j}_i:\ i\in\sI\paren{j}\right\}$ in \eqref{choose_j}, then one can choose $a_{i_1}$ or $a_{i_2}$ as $a\paren{j}$. However, the following proposition ensures that the  rium stock price is uniquely determined by our construction in Section 4, regardless of the possible different choices of the ordering $(a\paren{j})_{1\leq j\leq I}$.

\begin{proposition}
By the backward construction \eqref{def:sigma}-\eqref{various:def}, $\mu$ defined in \eqref{def:mu} is uniquely determined.
\end{proposition}
\begin{proof}
Suppose that $2\leq j \leq I-1$ and $\{i_1, i_2,...,i_m \}= \Argmax\left\{\eta\paren{j}_i:\ i\in\sI\paren{j}\right\}$ for $1\leq  m \leq {j}$ in the $j$-step of the backward induction in \eqref{def:sigma}-\eqref{various:def}. 
Without loss of generality, let $a\paren{j}=a_{i_1}$ and $\tau\paren{j}>0$. Then,
\begin{align}
\tau\paren{j}
    \begin{dcases}
    =\eta\paren{j}_{i}, & i \in \{i_1, i_2,...,i_m \} \\
    >\eta\paren{j}_{i}, & i \in \sI\paren{j}\setminus \{i_1, i_2,...,i_m \}
    \end{dcases}.
\end{align}
The above observation and Theorem~\ref{thm:tau-ordering} produce
\begin{align}
    \left| \left(a_i-\frac{a\paren{\geq j+1}_\Sigma}{I-j}\right) F(\tau\paren{j}) + \sum_{k=j+1}^{I-2} \frac{A\paren{k}}{I-k}\,F(\tau\paren{k}) \right| 
    \begin{dcases}
    = \lambda, & i \in \{i_1, i_2,...,i_m \} \\
    <\lambda,  & i \in \sI\paren{j}\setminus \{i_1, i_2,...,i_m \} 
    \end{dcases}. \label{form1}
\end{align}
Using the identity $\frac{a\paren{\geq j}_\Sigma}{I-j+1}= \frac{a\paren{\geq j+1}_\Sigma}{I-j}+ \frac{A\paren{j}}{I-j}$, we obtain that for any $a_i$,
\begin{equation}
\begin{split}\label{form2}
  & \left| \left(a_i-\frac{a\paren{\geq j+1}_\Sigma}{I-j}\right) F(\tau\paren{j}) + \sum_{k=j+1}^{I-2} \frac{A\paren{k}}{I-k}\,F(\tau\paren{k}) \right| \\
   &=   \left| \frac{I-j+2}{I-j+1} \left(a_i-\frac{a_i + a\paren{\geq j}_\Sigma}{I-j+2}\right) F(\tau\paren{j}) + \sum_{k=j}^{I-2} \frac{A\paren{k}}{I-k}\,F(\tau\paren{k}) \right|    .
\end{split}
\end{equation}
We combine \eqref{form1}, \eqref{form2}, and Theorem~\ref{thm:tau-ordering} to conclude that
\begin{align}
    \tau\paren{j}=\tau\paren{j-1}
    \begin{dcases}
   =\eta\paren{j-1}_{i}, & i \in \{i_2,...,i_m \} \\
    >\eta\paren{j-1}_{i}, & i \in \sI\paren{j-1}\setminus \{i_2,...,i_m \}
    \end{dcases}.
\end{align}
By above, without loss of generality, we set $a\paren{j-1}=a_{i_2}$. As an induction hypothesis, suppose that for $1\leq l\leq {m-1}$,
\begin{align}
a\paren{j-k+1}&=a_{i_k}, \quad 1\leq k \leq l,\\
    \tau\paren{j}&=\tau\paren{j-1}=\cdots = \tau\paren{j-l+1}
    \begin{dcases}
   =\eta\paren{j-l+1}_{i}, & i \in \{i_l,...,i_m \} \\
    >\eta\paren{j-l+1}_{i}, & i \in \sI\paren{j-l+1}\setminus \{i_l,...,i_m \}
    \end{dcases}. \label{tau_same}
\end{align}
Suitable rearrangement of \eqref{sum_identity} produces
\begin{align}
\sum_{k=j-l+1}^j \frac{A\paren{k}}{I-k} =\frac{a\paren{\geq j-l+1}_\Sigma}{I-j+l} - \frac{a\paren{\geq j+1}_\Sigma}{I-j}. \label{sum_identity2}
\end{align}
Using \eqref{tau_same} and \eqref{sum_identity2}, we obtain that for any $a_i$,
\begin{equation}
\begin{split}\label{form3}
  & \left| \left(a_i-\frac{a\paren{\geq j+1}_\Sigma}{I-j}\right) F(\tau\paren{j}) + \sum_{k=j+1}^{I-2} \frac{A\paren{k}}{I-k}\,F(\tau\paren{k}) \right| \\
   &=   \left| \frac{I-j+l+1}{I-j+l} \left(a_i-\frac{a_i + a\paren{\geq j-l+1}_\Sigma}{I-j+l+1}\right) F(\tau\paren{j}) + \sum_{k=j-l+1}^{I-2} \frac{A\paren{k}}{I-k}\,F(\tau\paren{k}) \right|    .
\end{split}
\end{equation}
We combine \eqref{form1}, \eqref{form3}, and Theorem~\ref{thm:tau-ordering} to conclude that
\begin{align}
    \tau\paren{j}&=\tau\paren{j-1}=\cdots = \tau\paren{j-l}
    \begin{dcases}
   =\eta\paren{j-l}_{i}, & i \in \{i_{l+1},...,i_m \} \\
    >\eta\paren{j-l}_{i}, & i \in \sI\paren{j-l}\setminus \{i_{l+1},...,i_m \}
    \end{dcases},
\end{align}
and without loss of generality, we set $a\paren{j-l}=a_{i_{l+1}}$. 

Therefore, by induction, we conclude that the agents $i_1,...,i_m$ stop trading at the same time $\tau\paren{j}$, regardless of the rearrangement choice 
$\{a\paren{j},a\paren{j-1},...,a\paren{j-m+1}\}$ of $\{a_{i_1},a_{i_2},...,a_{i_m}\}$.

To complete the proof of the uniqueness of $\mu$ in \eqref{def:mu}, it only remains to show that the rearrangement choice 
$\{a\paren{j},a\paren{j-1},...,a\paren{j-m+1}\}$ of $\{a_{i_1},a_{i_2},...,a_{i_m}\}$
does not affect the value of $\tau\paren{j-m}$. Indeed, this can be checked by the following observation:
\begin{align*}
 \sum_{k=j-m+1}^{I-2} \frac{A\paren{k}}{I-k}\,F(\tau\paren{k}) &=  \sum_{k=j+1}^{I-2} \frac{A\paren{k}}{I-k}\,F(\tau\paren{k}) + \left( \sum_{k=j-m+1}^j \frac{A\paren{k}}{I-k}\right) F(\tau\paren{j})\\
 &= \sum_{k=j+1}^{I-2} \frac{A\paren{k}}{I-k}\,F(\tau\paren{k}) +
  \left(  \frac{\sum_{l=1}^m a_{i_l}}{I-j+m}  - \frac{m \, a\paren{\geq j+1}_\Sigma}{(I-j+m)(I-j)}      \right) F(\tau\paren{j}),
\end{align*}
where the first equality is due to $\tau\paren{j}=...=\tau\paren{j-m+1}$, and the second equality is due to \eqref{sum_identity2} with $l=m$.  Obviously, the above expression does not depend on the rearrangement choice.
\end{proof}

We are primarily interested in how $\mu$ varies with the choice of $\lambda$.  In our model, $\mu$ depends on $\lambda$ in a complex, nonsuccinct way. The $\lambda$-dependence on $\mu$ can best be captured through $\mu$'s derivative, as described in Corollary~\ref{cor:lambda} below.  Corollary~\ref{cor:lambda} shows that the formula for the derivative of $\mu/\kappa$ does not change with $\lambda$, but the time intervals, on which the formulas apply, do change with $\lambda$.

All equilibrium outputs depend on $\lambda$, but $\lambda$ has been fixed up until now.  In what follows, we append $(\lambda)$ to our notation when we need to indicate that the market in question has a transaction cost proportion $\lambda$; e.g., $\mu(\lambda)$, $\tau\paren{j}(\lambda)$, etc.
\begin{corollary}\label{cor:lambda}
  Let $\lambda>0$ and suppose that $\gamma$ is differentiable on $[0,T]$.  Then,
  $$
    \frac{d}{dt}\left[\frac{\mu_t(\lambda)}{\kappa(t)}\right]
    =\begin{dcases}
    -\frac{\gamma'(t)a\paren{\geq j+1}_\Sigma}{I-j}, &t\in(\tau\paren{j}(\lambda),\tau\paren{j+1}(\lambda)), \ 0\leq j\leq I-2,\\
  -\frac{\gamma'(t)a\paren{\geq I-1}_\Sigma}{2},  &t\in(\tau\paren{I-1}(\lambda),T].
  \end{dcases}
  $$
In particular, suppose that the traders are TWAP traders with $\gamma^{\text{TWAP}}(t) := t/T$ for $t\in[0,T]$. Then for $\lambda, \lambda' >0$ and $0\leq j \leq I-1$, we have that $\frac{d}{dt}\left[\frac{\mu_t(\lambda)}{\kappa(t)}\right]$ for  $t\in\left(\tau\paren{j}(\lambda), \tau\paren{j+1}(\lambda)\right)$ is constant and agrees with $\frac{d}{dt}\left[\frac{\mu_t(\lambda')}{\kappa(t)}\right]$ for $t\in\left(\tau\paren{j}(\lambda'), \tau\paren{j+1}(\lambda')\right)$.
\end{corollary}
\begin{proof}
  By Proposition~\ref{prop:rank-based}, the rank-based ordering is determined independently of the choice of the transaction cost proportion.  Thus, taking derivatives using the definition of $\mu$ in \eqref{def:mu} yields the desired result.
\end{proof}

For TWAP traders, Corollary \ref{cor:lambda} shows that the two plots of $\mu$ have the same slope on the time intervals corresponding to when the same number of agents are trading.  However, the stop-trade times, and thus the number of agents trading, depend on $\lambda$.  Figure \ref{figure:two-drifts} (below) illustrates Corollary \ref{cor:lambda} by plotting two equilibrium drifts with different $\lambda$ parameters.  We consider the TWAP trading case with $\gamma(t)=\gamma^{\text{TWAP}}(t)= t/T$, $t\in[0,T]$ and have equilibrium inputs that are the same as in Section~\ref{section:ex} with $I=20$, $T=1$, $\kappa(t)=0.1$, $t\in[0,T]$, and trading targets given by \eqref{ai_data}.  The two graphs of equilibrium drift $\mu$ correspond to models with $\lambda=0.1$ and $\lambda=0.2$.
\begin{figure}[h]\label{figure:two-drifts}
\begin{center}
\includegraphics[scale=0.8]{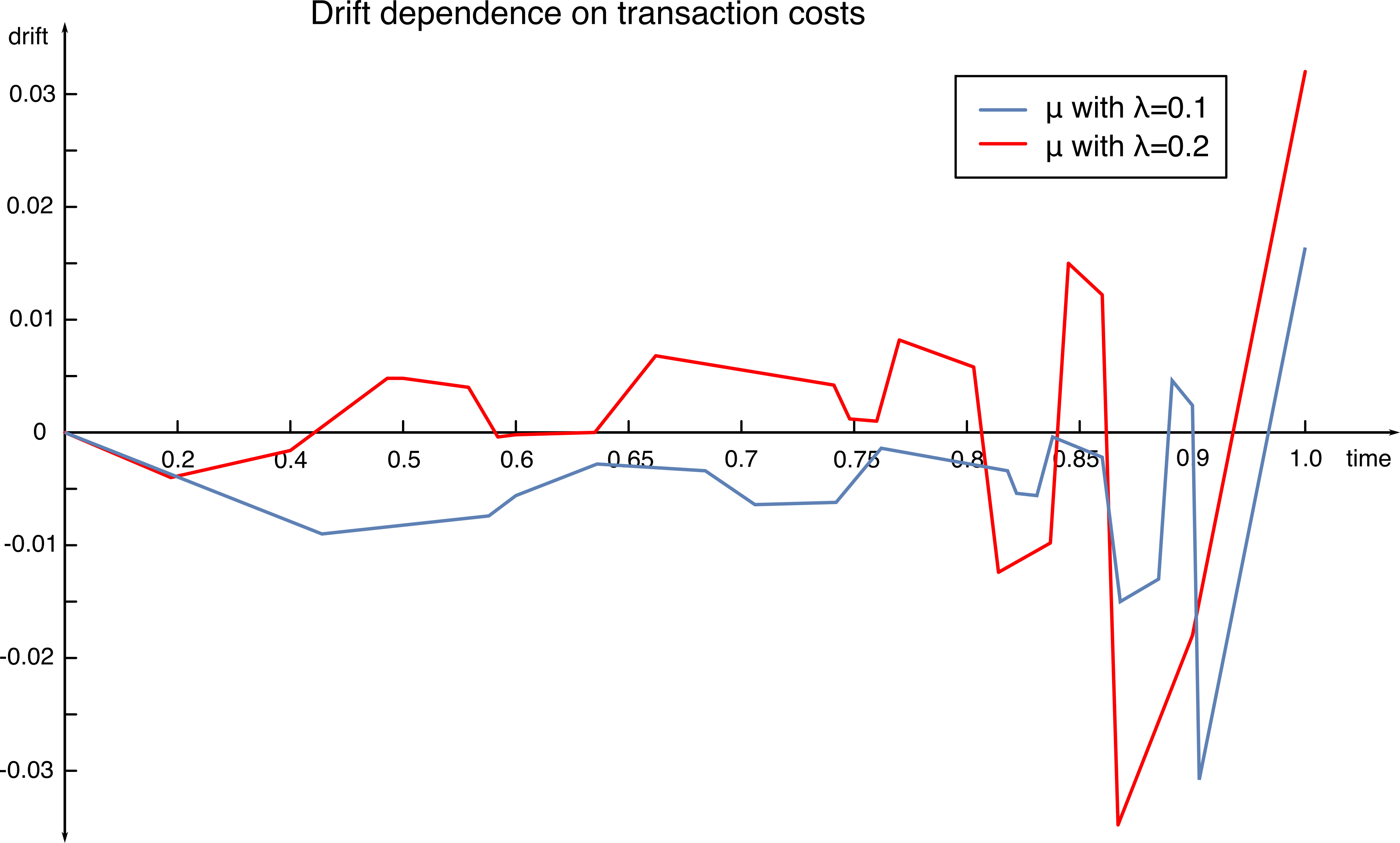}
\end{center}
  \caption{Two graphs of $\mu$ are plotted as functions of time $t\in[0,T]$ with parameters $I=20$, $T=1$, $n=0$, $\theta_{i,0-}=0$ for $1\leq i\leq I$, $\kappa(t)=0.1$ and $\gamma(t)=\gamma^{\text{TWAP}}(t)=t/T$, $t\in[0,T]$, and $(\tilde a_i)_{1\leq i\leq I}$ is given in \eqref{ai_data}.  The transaction costs $\lambda$ differs between the two $\mu$ plots by $\lambda=0.1$ and $\lambda=0.2$.
 The horizontal axis is suitably scaled to better visualize the changes in the graphs of $\mu$.
 }\label{figure:two-drifts}
\end{figure}

By inspecting the graphs plotted in Figure~\ref{figure:two-drifts}, we notice that the drifts have the same slope at time $T$.  (Note:  the time scale of Figure~\ref{figure:two-drifts} changes at time $0.9$, and so the graph of $\mu$ with $\lambda=0.2$ appears to change slope at time $0.9$, though the slope is actually constant.)  Tracing back further in time, $\mu$ with $\lambda=0.1$ and $\lambda=0.2$ have the same slopes on the intervals corresponding to the same number of agents trading.  

Higher levels of transaction costs lead to earlier (lower) stop-trade times, which lead to a wider range of possible $\mu$ values for higher $\lambda$.  A similar phenomenon appeared Gonon et.\,al.~\cite{GMKS21MF}.  In both this work and Gonon et.\,al.~\cite{GMKS21MF}, higher $\lambda$ leads to less trading, which results in changes in the equilibrium drift.  In Gonon et.\,al.~\cite{GMKS21MF}, the equilibrium drift also has a wider range of values and activity as $\lambda$ increases, due to an increase in the range of a doubly-reflected Brownian motion state process that drives the drift.

The impact of transaction costs can also be seen through the equilibrium stock price.  Interestingly, the impact is not monotone but instead is piecewise linear.  Figure~\ref{figure:S0} (below) plots the impact of $\lambda$ on $S_0$.  The equilibrium inputs are the same as in Section~\ref{section:ex} with $I=20$, $T=1$, $\kappa(t)=0.1$ and $\gamma(t) = \gamma^{\text{TWAP}}(t)=t/T$, $t\in[0,1]$, and trading targets given by \eqref{ai_data}.  We vary $\lambda$ in the range $(0, 2]$.
\begin{figure}[h]
\begin{center}
\includegraphics[scale=0.8]{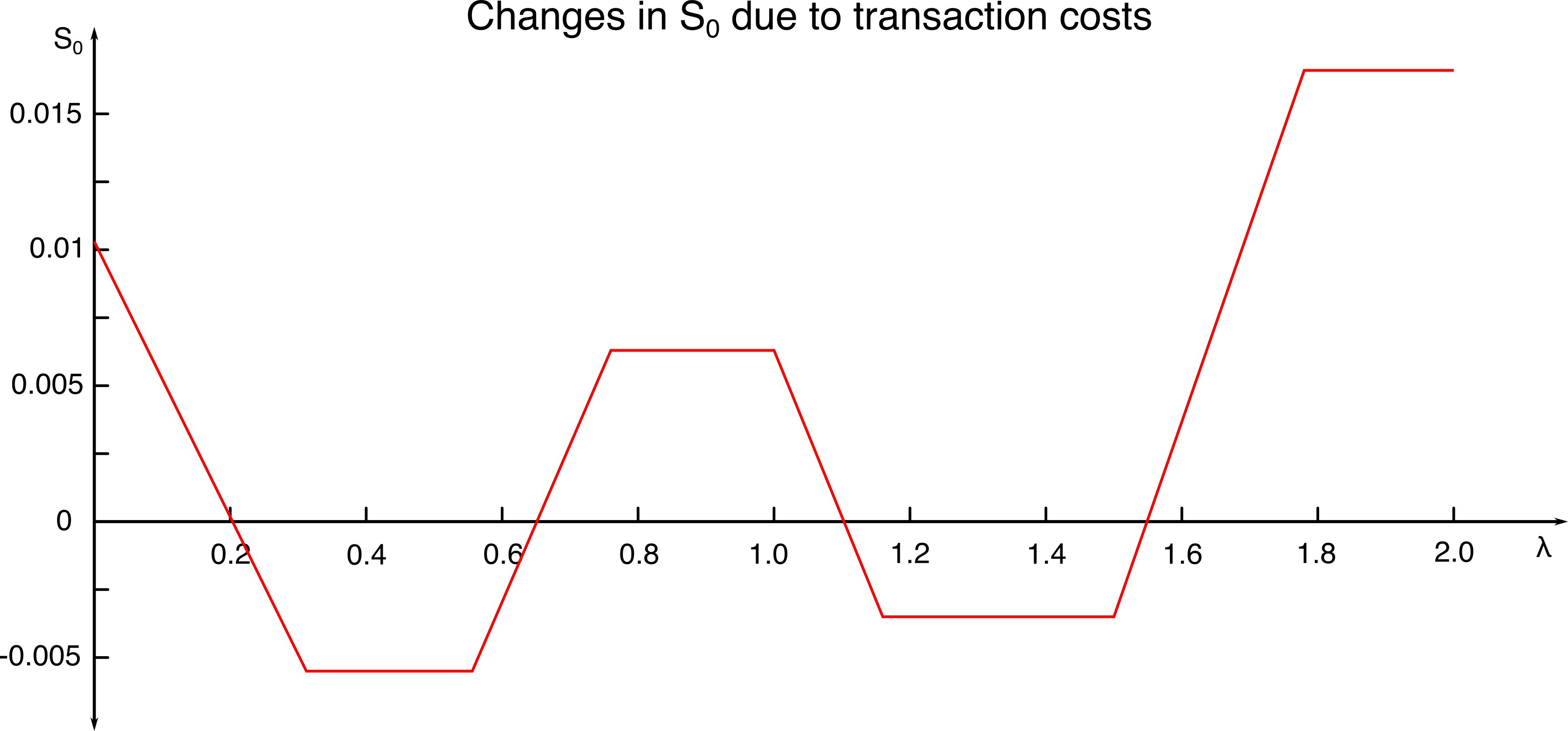}
\end{center}
  \caption{The value of $S_0$ is plotted as a function of $\lambda\in (0, 2]$. The parameters are $I=20, \, T=1, \, \lambda=0.2,\, n=0,\, \theta_{i,0-}=0$ for $1\leq i \leq I$, $\kappa(t)=0.1$ and $\gamma(t)=\gamma^{\text{TWAP}}(t)=t/T$ for $t\in [0,T]$, and $(\tilde a_i)_{1\leq i \leq I}$ is given in \eqref{ai_data}.
 }
 \label{figure:S0}
\end{figure}

Why does the plot of $S_0$ in Figure~\ref{figure:S0} have a piecewise linear shape?  
Corollary~\ref{cor:S0} (below) provides a formula for $S_0$ with insights into Figure~\ref{figure:S0}'s shape.

\begin{corollary} \label{cor:S0}
We have
$$
  S_0 = \E[D] +\int_0^T \frac{\kappa(t)\gamma(t)a\paren{\geq 1}_\Sigma}{I} dt
    - \sum_{j=1}^{I-2}\frac{A\paren{j}F(\tau\paren{j})}{I-j},
$$
and the map $\lambda \mapsto S_0 (\lambda)$ is piecewise linear.
\end{corollary}
\begin{proof}
We observe that for $0\leq j\leq I-2$ and $t\in[0,T]$,
\begin{align}
\frac{\gamma(t) a\paren{\geq j+1}_\Sigma}{I-j} + \sum_{k=1}^j \frac{\gamma(\tau\paren{k})A\paren{k}}{I-k}&=\gamma(t) \left( \frac{a\paren{\geq j+1}_\Sigma}{I-j} + \sum_{k=1}^j \frac{A\paren{k}}{I-k}   \right) -\sum_{k=1}^{j}\frac{A\paren{k}}{I-k}
  \left(\gamma(t)-\gamma(\tau\paren{k})\right) \nonumber \\
&=\frac{\gamma(t)a\paren{\geq 1}_\Sigma}{I}
  -\sum_{k=1}^{j}\frac{A\paren{k}}{I-k}
  \left(\gamma(t)-\gamma(\tau\paren{k})\right), \label{id_for_S0}
\end{align}
where the second equality above is due to \eqref{sum_identity2} with $l=j$.
Using the definition of $\mu$ in \eqref{def:mu} and the identity \eqref{id_for_S0}, we have
\begin{align}
  \mu_t = \kappa(t)\left(-\frac{\gamma(t)a\paren{\geq 1}_\Sigma}{I}
  +\sum_{k=1}^{j\wedge (I-2)}\frac{A\paren{k}}{I-k}
  \left(\gamma(t)-\gamma(\tau\paren{k})\right)\right), \label{mu_new}
\end{align}
for $t\in[0,T]$, where either $t\in[\tau\paren{j},\tau\paren{j+1})$, for some $0\leq j\leq I-2$, or $t\in[\tau\paren{I-1},T]$ for $j=I-1$.

Therefore, the definition of $S$ in \eqref{def:S} and the expression of $\mu$ in \eqref{mu_new} produce
\begin{align*}
  S_0 &- \E[D] = -\int_0^T \mu_t dt \\
 &= - \sum_{j=0}^{I-2}\int_{\tau\paren{j}}^{\tau\paren{j+1}} \mu_t dt - \int_{\tau\paren{I-1}}^{T} \mu_t dt \\
  &= -\sum_{j=0}^{I-2} \int_{\tau\paren{j}}^{\tau\paren{j+1}} \kappa(t)\left(-\frac{\gamma(t)a\paren{\geq 1}_\Sigma}{I}
  +\sum_{k=1}^{j}\frac{A\paren{k}}{I-k}
  \left(\gamma(t)-\gamma(\tau\paren{k})\right)\right)dt \\
  & \quad\quad- \int_{\tau\paren{I-1}}^T \kappa(t)\left(-\frac{\gamma(t)a\paren{\geq 1}_\Sigma}{I}
  +\sum_{k=1}^{I-2}\frac{A\paren{k}}{I-k}
  \left(\gamma(t)-\gamma(\tau\paren{k})\right)\right)dt\\
  &= \int_0^T \frac{\kappa(t)\gamma(t)a\paren{\geq 1}_\Sigma}{I}dt
  - \sum_{k=1}^{I-2} \frac{A\paren{k}}{I-k}\int_{\tau\paren{k}}^T \kappa(t)\left(\gamma(t)-\gamma(\tau\paren{k})\right)dt,
\end{align*}
where we interchange the order of the summations to obtain the last equality.

Finally, the expression above and Proposition~\ref{prop:rank-based} imply that the map $\lambda \mapsto S_0 (\lambda)$ is piecewise linear.
\end{proof}

\bibliographystyle{plain}
\bibliography{finance_bib}

\end{document}